\documentclass[11pt]{amsart}

\usepackage{amssymb,amsthm,amsmath}
\usepackage[numbers,sort&compress]{natbib}
\usepackage{color}
\usepackage{graphicx}
\usepackage{tikz}
\usepackage{amssymb,amsthm,amsmath}
\usepackage[numbers,sort&compress]{natbib}
\usepackage{color}
\usepackage{graphicx}
\usepackage{tikz}
\usepackage{comment}

\title[ ]{ Universal reflective-hierarchical structure of quasiperiodic
  eigenfunctions and sharp spectral transition in phase  }
\author{ Svetlana Jitomirskaya}
\address[ Svetlana Jitomirskaya]{ Department of Mathematics, University of California, Irvine, California 92697-3875, USA}
\email{szhitomi@math.uci.edu}

\author{Wencai Liu}
\address[Wencai Liu]{Department of Mathematics, University of California, Irvine, California 92697-3875, USA}\email{liuwencai1226@gmail.com}

\theoremstyle{plain}
\newtheorem{theorem}{Theorem}[section]
\newtheorem{corollary}[theorem]{Corollary}
\newtheorem{lemma}[theorem]{Lemma}

\newcommand{\R}{\mathbb{R}}
\newcommand{\Z}{\mathbb{Z}}
\theoremstyle{definition}
\newtheorem{definition}[theorem]{Definition}

\newtheorem{remark}[theorem]{Remark}

\begin{document}


\begin{abstract}
We prove sharp  spectral transition in the arithmetics of phase between localization and singular
continuous spectrum for  Diophantine  almost
Mathieu operators. We also
determine exact exponential asymptotics of eigenfunctions and
of corresponding transfer matrices  throughout the localization
region. This uncovers a universal structure in their behavior governed
by the exponential phase resonances. The structure features a new type
of hierarchy, where self-similarity holds upon alternating reflections.

\end{abstract}
\maketitle

\section{Introduction}

Unlike random, one-dimensional quasiperiodic operators feature
spectral transitions with changes of parameters. The transitions
between absolutely continuous and singular spectrum are governed by
vanishing/non-vanishing of the Lyapunov exponent \cite{kot}. In the regime of positive Lyapunov exponents (also called
supercritical in the analytic case, with the name inspired by the almost Mathieu operator)  there are also more delicate transitions: between localization
(point spectrum with exponentially decaying eigenfunctions) and
singular continuous spectrum. They are governed by the resonances: eigenvalues
of box restrictions that are too close to each other in relation to the
distance between the boxes, leading to small denominators in various expansions. Localization is said to be a game of
resonances, a statement attributed to P. Anderson and Ya. Sinai (e.g. \cite{gs}). Indeed, all known proofs of localization, starting with
Fr\"ohlich-Spencer's multi-scale analysis \cite{fs} are based, in one
way or another, on avoiding
resonances and removing resonance-producing parameters, while all known proofs of singular continuous spectrum and
even some of the absolutely continuous one \cite{aiz}
are based on showing their
abundance.

For quasiperiodic operators, one category of resonances are the ones
determined entirely by the frequency. Indeed, for smooth potentials,
large coefficients in the continued fraction expansion of the
frequency lead
to almost repetitions and thus resonances, regardless of the values of
other parameters.  Such resonances were first
understood and exploited to show
singular continuous spectrum for Liouville frequencies  in
\cite{as, as1}, based on \cite{gordon} \footnote{According to \cite{mar},
the fact that the Diophantine properties of the frequencies
should play a role was first observed in \cite{sarnak}.}. The strength of frequency
resonances is measured by the arithmetic parameter
\begin{equation}\label{beta}
\beta(\alpha)=\limsup_{k\to \infty} -\frac{\ln ||k\alpha||_{\R/\Z}}{|k|}
\end{equation}
where $||x||_{\R/\Z}=\inf_{\ell\in \Z}|x-\ell|$.
Another class of
resonances,
appearing for all {\it even} potentials, was discovered in \cite{js}, where it
was shown for the first time that the arithmetic properties of
the phase also play a role and may lead to singular continuous spectrum
even for the Diophantine frequencies. Indeed, for even potentials, phases with
almost symmetries lead to resonances, regardless of the
values of other parameters.\footnote{Symmetry based resonances were first observed in
  \cite{fsw} for the almost Mathieu operator in the perturbative regime.}\label{f} The strength of phase
resonances is measured by the arithmetic parameter
\begin{equation}\label{G.delta}
\delta(\alpha,\theta)=\limsup_{k\to\infty} -\frac{\ln ||2\theta+ k\alpha||_{\R/\Z}}{|k|}
\end{equation}

In both these cases, the strength of the
resonances is in competition with the exponential growth controlled by the
Lyapunov exponent.  It was conjectured in 1994 \cite{conj}
that for the almost Mathieu family- the prototypical quasiperiodic
operator - the two above types of resonances are the only ones
that appear (as is the case in the perturbative regime of \cite{fsw}), and the competition between the Lyapunov growth and
resonance strength  resolves, in both cases, in a sharp way. The
frequency half of the conjecture was recently  solved
\cite{ayz,jl1}. 
In this paper we present the solution of  the phase
half. Moreover, our proof of the pure point part of the conjecture
uncovers  a universal structure of the eigenfunctions throughout the
entire pure point spectrum regime, which, in presence of exponentially strong
resonances, demonstrates a new phenomenon that we call a {\it reflective
  hierarchy}, when the eigenfunctions feature self-similarity upon
proper reflections. This phenomenon was not even previously described
in the (vast) physics literature. 
This paper is, in some sense, dual to the recent work  \cite{jl1}. While the
universal hierarchical structure governed by the frequency resonances
discovered in \cite{jl1} is conjectured to hold, for a.e. phase,
throughout the entire class of analytic potentials, the structure
discovered here requires evenness of the function defining the
potential, and moreover, in general, resonances of other types may
also be present. However, we conjecture that
for general even analytic potentials for a.e. frequency only finitely many other exponentially strong
resonances will appear, thus the structure described in this paper will hold for
the corresponding class, with the $\ln\lambda$ replaced by the Lyapunov exponent $L(E)$ throughout.

  The almost Mathieu operator (AMO) is the (discrete) quasiperiodic   Schr\"{o}dinger operator  on  $   \ell^2(\mathbb{Z})$:
 \begin{equation}\label{Def.AMO}
 (H_{\lambda,\alpha,\theta}u)(n)=u({n+1})+u({n-1})+ 2\lambda v(\theta+n\alpha)u(n),  \text{ with }  v(\theta)=2\cos2\pi \theta,
 \end{equation}
where $\lambda$ is the coupling, $\alpha $ is the frequency, and
$\theta $ is the phase.

It is the central quasiperiodic model due to coming from physics and attracting continued
interest there. It first appeared in Peierls \cite{peierls1933theorie}, and arises as related, in two different ways, to a
two-dimensional electron subject to a perpendicular magnetic field. It
plays the central role in
the Thouless et al theory of the integer quantum Hall effect. For further background, history,
and surveys of results see  \cite{jitomirskaya2015dynamics,dam,bbook,sim60,last} and references therein.


Frequency $\alpha$ is called Diophantine if there exist  $\kappa>0$ and $\tau>0$ such
that  for $k\neq 0$,
\begin{equation}\label{Diokt}
    ||k\alpha||_{\R/\Z}\geq \frac{\tau}{|k|^{\kappa}}.
\end{equation}
From here on, unless otherwise noted, we will
always assume $\alpha$ is Diophantine.  When we need to refer to (\ref{Diokt}) in a non-quantitative way we will sometimes call it  the Diophantine condition (DC) on $\alpha.$ \footnote{It is rather
  straightforward to extend all the results to the case
  $\beta(\alpha)=0,$ without any changes in formulations. We present
  the proof under the condition (\ref{Diokt}) just for a slight
  simplification of some arguments.}

  Operator $H$ is said to have
Anderson localization if it has pure point spectrum with exponentially
decaying eigenfunctions.



We have
\begin{theorem}\label{Maintheorem}

\begin{enumerate}
\item [1.] 
$H_{\lambda,\alpha,\theta}$ has Anderson localization if $|\lambda|> e^{\delta(\alpha,\theta)}$,
\item [2.]
$H_{\lambda,\alpha,\theta}$ has purely singular  continuous
  spectrum  if $1<|\lambda|< e^{ \delta(\alpha,\theta)}$.
\item [3.] 
$H_{\lambda,\alpha,\theta}$  has
  purely abolutely continuous
  spectrum  if $|\lambda|< 1$.
\end{enumerate}
\end{theorem}
{\bf Remark}
\begin{enumerate}
\item We will prove part 2 for all irrational $\alpha, $ and general Lipshitz $v$ in (\ref{Def.AMO}), see Theorem \ref{general}.
\item Part 3 is known for all $\alpha,\theta$  \cite{av} and is included here for completeness.
\item Parts 1 and 2 of Theorem \ref{Maintheorem} verify the phase half of
  the conjecture in \cite{conj}, as stated there. The frequency half was recently proved
in \cite{ayz,jl1}.
\end{enumerate}

Singular continuous spectrum was first established for $1<|\lambda|<
e^{ c\delta(\alpha,\theta)}$, for sufficiently small $c$
\cite{js}. One can see that even with tight upper semicontinuity
bounds the argument of \cite{js} does not work for $c>1/4.$  Here we
introduce new ideas to remove the factor of $4$ and approach the
actual threshold.

Anderson localization for Diophantine $\alpha$ and
$\delta(\alpha,\theta)=0$ was proved in \cite{MR1740982}. The argument was theoretically extendable to $|\lambda|>
e^{ C\delta(\alpha,\theta)}$ for a large $C$ but not
beyond. Therefore, the case of $\delta(\alpha,\theta)>0$ was
completely open before. In fact, the localization method of \cite{MR1740982}
could not deal with exponentially strong resonances. The first way to
handle exponentially strong {\it frequency} resonances was developed
in \cite{MR2521117}. It was then pushed to the technical limits in \cite{MR3340177} but that method could not
approach the threshold. An important technical achievement of \cite{jl1}
was to develop a way to handle frequency resonances that works up to
the very transition and leads to sharp bounds. In this paper we
develop the first, and at the same time the sharp, way to treat exponential {\it phase}
resonances.

We
borrow two basic technical ingredients from prior work, that we abstract out as
Theorems \ref{universalth} and \ref{blockth} which we prove in the
appendices. Otherwise, since
frequency and phase resonances are fundamentally different in nature
(one is based on the repetitions and the other on reflections), the
specific techniques and constructions required to achieve sharp
results both on the point/upper bounds and
singular continuos/lower bounds sides are completely different.

Recently, it became possible to prove pure point spectrum in a
non-constructive way,  avoiding the
localization method, using instead
reducibility  for the dual model \cite{ayz} (see also
\cite{jk}) as was first done,
in the perturbative regime in \cite{bel}. Coupled with recent   arguments \cite{a,afk,hy,yz} that allow to conjugate the
global transfer-matrix cocycle into the local almost reducibility regime \footnote{For the Diophantine case this is the
Eliasson's regime \cite{el}} and proceed by almost
reducibility,
this offers a powerful technique that
led to a solution of the measure theoretic version of the frequency
part of the conjecture of \cite{conj} by Avila-You-Zhou in \cite{ayz} and a
corresponding sharp result  for the supercritical regime in the
extended Harper's model \cite{hj}. However, we note that proofs by
dual reducibility inherently lose the control over phases (thus can
only be measure theoretic), and therefore cannot approach the
transitions in phase.

Our proof of localization is based on determining the {\it exact  asymptotics} of the generalized eigenfunctions in the regime $ |\lambda|>e^{\delta(\alpha,\theta)}$.

For any  $\ell$,
              let  $x_0$ (we can choose any one if $x_0$ is not unique) be such that
              \begin{equation*}
             |\sin\pi(2\theta+x_0\alpha)|  = \min_{|x|\leq 2|\ell|}|\sin\pi(2\theta+x\alpha)|.
              \end{equation*}
               Let $\eta=0$ if $2\theta+x_0\alpha\in \mathbb{Z}$, otherwise  let $\eta\in (0, \infty)$ be given by the following equation,
                \begin{equation}\label{eta}
                   |\sin\pi(2\theta+x_0\alpha)|=e^{-\eta |\ell|}.
                \end{equation}
                    Define $f: \mathbb{Z}\rightarrow \mathbb{R}^+$ as follows.

                    Case 1: $x_0\cdot \ell \leq  0$. Set
                     \begin{equation*}
                        f(\ell)  = e^{-|\ell|\ln|\lambda|}  .
                    \end{equation*}

                    Case 2. $x_0\cdot \ell >  0$. Set

                    \begin{equation*}
                        f(\ell)  = e^{-(|x_0|+|\ell-x_0|) \ln|\lambda|} e^ {\eta |\ell|}  +e^{-|\ell|\ln|\lambda|}  .
                    \end{equation*}

We say that $\phi$ is a generalized  eigenfunction of $H$ with generalized
eigenvalue $E$, if
   \begin{equation} \label{shn}
     H\phi=E\phi  ,\text{ and }  |\phi(k)|\leq \hat{C}(1+|k|).
  \end{equation}

   For a fixed generalized eigenvalue $E$ and  corresponding generalized eigenfunction  $\phi$ of $H_{\lambda,\alpha,\theta}$, let $ U(\ell) =\left(\begin{array}{c}
                             \phi(\ell)\\
                            \phi({\ell-1})
                          \end{array}\right)
                      $.
We have

                     \begin{theorem}\label{Maintheoremdecay}
                    Assume $\ln |\lambda| >\delta(\alpha,\theta)$. If  $E$ is a generalized eigenvalue  and  $\phi$ is
                     the  corresponding generalized eigenfunction  of $H_{\lambda,\alpha,\theta}$, then for any $\varepsilon>0$, there exists $K$ such that for any $|\ell|\geq K$,  $U(\ell)$ satisfies
                     \begin{equation}\label{G.Asymptotics}
                      f(\ell)e^{-\varepsilon|\ell|} \leq ||U(\ell)||\leq f(\ell)e^{\varepsilon|\ell|}.
                       \end{equation}
                            In particular,  the eigenfunctions decay at the rate $\ln |\lambda| -\delta(\alpha,\theta)$.
                     \end{theorem}
             {\bf Remark}\label{rem}

              \begin{itemize}
              \item For $\delta=0$ we have that for any $\varepsilon>0$,
              \begin{equation*}
             e^{-(\ln|\lambda|+\varepsilon)|\ell|}  \leq  f(\ell)\leq e^{-(\ln|\lambda|-\varepsilon)|\ell|}.
              \end{equation*}
              This implies that the eigenfunctions decay precisely at the rate of Lyapunov exponent $\ln |\lambda|$.
              \item
              For $\delta>0$, by the definition of $\delta$ and $f$, we have for any $\varepsilon>0$,
              \begin{equation}
                f(\ell)\leq e^{-(\ln|\lambda|-\delta-\varepsilon)|\ell|}.
              \end{equation}
              \item
              By the definition of $\delta$ again, there exists a subsequence $\{\ell_i\} $ such that
              \begin{equation*}
                 |\sin\pi(2\theta+\ell_i \alpha)|\leq e^{-(\delta-\varepsilon) |\ell_i|}.
              \end{equation*}
              By the DC  on $\alpha$, one has that
              \begin{equation*}
                |\sin\pi(2\theta+\ell_i\alpha)|  = \min_{|x|\leq 2|\ell_i|}|\sin\pi(2\theta+x\alpha)|.
              \end{equation*}
              Then
               \begin{equation}\label{locmax}
               f(\ell_i)\geq e^{-(\ln|\lambda|-\delta+\varepsilon)|\ell_i|}.
               \end{equation}
               This implies the eigenfunctions decay precisely at the rate $\ln |\lambda| -\delta(\alpha,\theta)$.
                 \item
                 If $x_0$ is not unique, by the DC on $\alpha$, we must have that $\eta$ is arbitrarily small.
                 Then
                 \begin{equation*}
                  e^{-(\ln|\lambda|+\varepsilon)|\ell|}  \leq ||U(\ell)||\leq e^{-(\ln|\lambda|-\varepsilon)|\ell|}.
                 \end{equation*}

              \end{itemize}

The behavior
described in Theorem  \ref{Maintheoremdecay}  happens around arbitrary
point. This, coupled with effective control of parameters at the local
maxima, allows to uncover the self-similar nature of the
eigenfunctions.  Hierarchical behavior of solutions,
despite significant numerical studies and even a discovery of  Bethe Ansatz
solutions \cite{wieg,aw}  has  remained an important open challenge even at
the physics level.
In paper \cite{jl1} we obtained  universal hierarchical structure  of
the  eigenfunctions  for all frequencies  $\alpha $ and phases with
$\delta(\alpha,\theta)=0$. In studying the eigenfunctions of
$H_{\lambda,\alpha,\theta}$ for $\delta(\alpha,\theta)>0$ we obtain a
different kind of universality throughout the pure point spectrum
regime, which features a self-similar hierarchical structure 
upon proper {\it reflections}.

\vskip.2in

Assume phase $\theta$ satisfies $0<\delta(\alpha,\theta)<\ln \lambda$. Fix $0<\varsigma <\delta(\alpha,\theta).$

 Let $k_0$ be a global maximum of eigenfunction $\phi$.\footnote{Can take any one if there are several.}
Let $K_i$ be the positions of exponential resonances of the phase
$\theta'=\theta +k_0\alpha $ defined by
\begin{equation}\label{resonantKi}
  ||2\theta+(2k_0+K_i)\alpha||_{\R/\Z}\leq  e^{-\varsigma |K_i|},
\end{equation}

This means that $|v(\theta'+\ell\alpha)-v(\theta'+(K_i-\ell)\alpha)|\leq
Ce^{-\varsigma |K_i|}$, uniformly in $\ell,$ or, in other words, the
potential $v_n=v(\theta+n\alpha)$ is $ e^{-\varsigma |K_i| }$-almost symmetric with respect to
$(k_0+K_i)/2.$

Since $\alpha$ is Diophantine, we have
\begin{equation}\label{KiKi-1}
   |K_i|\geq c e^{ c|K_{i-1}|},
\end{equation}
where $c$ depends on $\varsigma$ and $\alpha$ through the Diophantine constants $\kappa,\tau.$ On the other hand, $K_i$ is necessarily an infinite sequence.

Let $\phi$ be an  eigenfunction, and   $ U(k) =\left(\begin{array}{c}
        \phi(k)\\
       \phi({k-1})
     \end{array}\right)
 $.
We say $k$ is a local $K$-maximum if $||U(k)||\geq ||U(k+s)||$ for all $s-k\in
[-K,K]$.

We first describe the hierarchical structure of local maxima
informally. There exists a {\it constant} $\hat {K}$ such that there
is a local $cK_j$-maximum $b_{j}$ within distance $\hat {K}$ of
each resonance $K_j$. The exponential
behavior of the eigenfunction
in the local $cK_j$-neighborhood of each such local maximum, normalized
by the value at the local maximum is given by the {\it reflection} of
$f$. Moreover, this describes the entire collection of local maxima of
depth $1$, that is $K$ such that $K$ is a $cK$-maximum. Then we have a similar picture in the vicinity of $b_{j}:$
there are local $cK_i$-maxima $b_{j,i}, i<j,$ within distance $\hat {K}^2$ of
each $K_j-K_i.$  The exponential
(on the $K_i$ scale) behavior of the eigenfunction
in the local $cK_i$-neighborhood of each such local maximum, normalized
by the value at the local maximum is given by
$f$. Then we get the next level maxima
$b_{j,i,s}, s<i$ in the $\hat {K}^3$-neighborhood of $K_j-K_i+K_s$ and reflected
behavior around each, and so on, with reflections alternating with
steps. At the end we obtain a
  complete hierarchical structure of local maxima that we denote by
  $b_{j_0,j_1,\ldots,j_s}, $ with each
  ``depth $s+1$" local maximum $b_{j_0,j_1,\ldots,j_{s}}$ being in the corresponding
  vicinity of the ``depth $s$" local maximum
  $b_{j_0,j_1,\ldots,j_{s-1}}\approx k_0+\sum_{i=0}^{s-1} (-1)^{i}K_{j_i} $ and with universal
  behavior at the corresponding scale around each. The quality of the
  approximation of the position of the next maximum gets lower with each level of
  depth, with $b_{j_0,j_1,\ldots,j_{s-1}}$ determined with $\hat {K}^s$
  precision, thus it presents an accurate picture as long as $ K_{j_s}  \gg
 \hat{K}^s.$

We now describe the hierarchical structure precisely.

\begin{theorem}\label{Universalend}
Assume sequence $K_i$ satisfies (\ref{resonantKi}) for some $\varsigma>0$.
Then
there exists   $\hat {K}  (\alpha,\lambda,\theta,\varsigma)<\infty$\footnote{$\hat {K}$ depends on $\theta$ through $2\theta+k\alpha$, see \eqref{G.delta}.}  such that for any  $j_0>j_1>\cdots>j_k\geq 0$ with $K_{j_k}\geq   \hat{K}^{k+1}$,
for each $0\leq s\leq k$ there exists a
local $\frac{\varsigma}{2\ln\lambda}K_ {j_s}$-maximum\footnote{Actually, it can be a local $(\frac{\varsigma}{\ln\lambda}-\varepsilon) K_ {j_s}$-maximum for any $\varepsilon>0$.} $b_{{ j_0},{ {j_1}},...,{ {j_s}}}$
such that the
following holds:
\begin{description}
  \item[I]   $|b_{{ j_0},{ {j_1}},...,{ {j_s}}}-k_0-
  \sum_{i=0}^s (-1)^{i} K_{j_i}|\leq   \hat{K}^{s+1}.$
  \item[II]For any $\varepsilon>0$,  if $ C\hat{K}^{k+1}\leq |x-b_{{ j_0},{ {j_1}},...,{ {j_k}}}|\leq \frac{\varsigma}{4\ln\lambda}|K_{ {j_k}}|$, where $C$ is a large constant depending on $ \alpha,\lambda,\theta,\varsigma$ and $\varepsilon$, then  for each
    $s=0,1,...,k,$

\begin{equation}\label{G.add1local}
 f((-1)^{s+1}x_s)e^{-\varepsilon|x_s|} \leq \frac{||U(x)||}{||U(b_{{ j_0},{ {j_1}},...,{ {j_s}}})||}\leq f((-1)^{s+1}x_s)e^{\varepsilon|x_s|},
\end{equation}
where $x_s=x- b_{{ j_0},{ {j_1}},...,{ {j_s}}} $.
\end{description}

\end{theorem}

\begin{remark}
Actually \eqref{G.add1local} holds for $x$ with $ C\hat{K}^{k+1}\leq |x-b_{{ j_0},{ {j_1}},...,{ {j_k}}}|\leq (\frac{\varsigma}{2\ln\lambda}-\varepsilon)|K_{ {j_k}}|$ for any $\varepsilon>0$.
\end{remark}

Thus the behavior of $\phi(x)$ is described by the same universal $f$
in each $\frac{\varsigma}{2\ln\lambda}K_ {j_s}$ window around the
corresponding local maximum $b_{{ j_0},{ {j_1}},...,{ {j_s}}}$
after alternating reflections. The positions of the local maxima in
the hierarchy are determined up to errors that at all but possibly
the last step are superlogarithmically small in $K_ {j_s}.$  We call
such a structure {\it reflective hierarchy}.

We are not aware of previous results describing the structure of
eigenfunctions for Diophantine $\alpha$ (The structure in the regime
$\beta >0$ is described in \cite{jl1}). Certain results indicating the
hierarchical structure in the corresponding
semi-classical/perturbative regimes were previously obtained in the
works of Sinai, Helffer-Sjostrand, and Buslaev-Fedotov (see \cite{helf,busl,sin}, and also \cite{zhi} for another model
). We were
also informed \cite{GP} that for strongly Diophantine $\alpha$ the fact that many
eigenfunctions of box restrictions for analytic $v$ in (\ref{Def.AMO})
can only ``bump up'' at resonances, can be obtained from the
avalanche principle expansion of the determinants, an important
method developed in \cite{gs}.

\begin{center}
\begin{tikzpicture}[thick, scale=1.8]
\node [below] at (3.0,5.6){reflective self-similarity  of  an eigenfunction};
\draw[dashed](0,0)--(0,4.9);
\draw [->](-1,0)--(9,0);
\draw[dashed](1.97,4)--(1.97,0);
\draw[dashed](8,2.2)--(8,0);
\draw[dashed](7,1.8)--(7,0);
\draw[dashed](0,4.9)--(-1,4.9);
\draw[dashed](-1,4.9)--(-1,0);
\node [left] at (-0.5,2){I};
\node [left] at (1.0,2.0){II};
\node [left] at (7.6,1.0){II$^\prime$};
\draw[dashed](8,2.2)--(8.8,2.2);
\draw[dashed](8.8,2.2)--(8.8,0);
\node [left] at (8.5,1.0){I$^\prime$};
\node [above] at (0,4.9){Global maximum};
\node [below] at (1.85,0){$K_{j_1}$};
\node [below] at (1.85,4.3){$b_{j_1}$};
\node [below] at (8,0){$K_{j_0}$};
\node [below] at (8,2.6){$b_{j_0}$};
\node [below] at (7,2.2){$b_{{j_0},{j_1}}$};
\node [below] at (0,0){0};
\draw plot [smooth] coordinates {(-0.5,4.5)(0,4.9)(0.5,4.5)(1.3,3)(2,4)(3,3)};
\draw [dashed] plot [smooth] coordinates {(3,3)(5,1)(6,1.3)};

\draw  plot [smooth] coordinates {(6,1.3)(6.55,1.6)(7,1.8)(7.3,1.5)(8,2.2)(8.5,1.8)};
\draw[dashed,->](7.4,1.3)--(7,1.3);
\draw[dashed,->](7.6,1.3)--(8,1.3);
\node [below] at (7.5,1.45){$K_{j_1}$};
\end{tikzpicture}
\end{center}

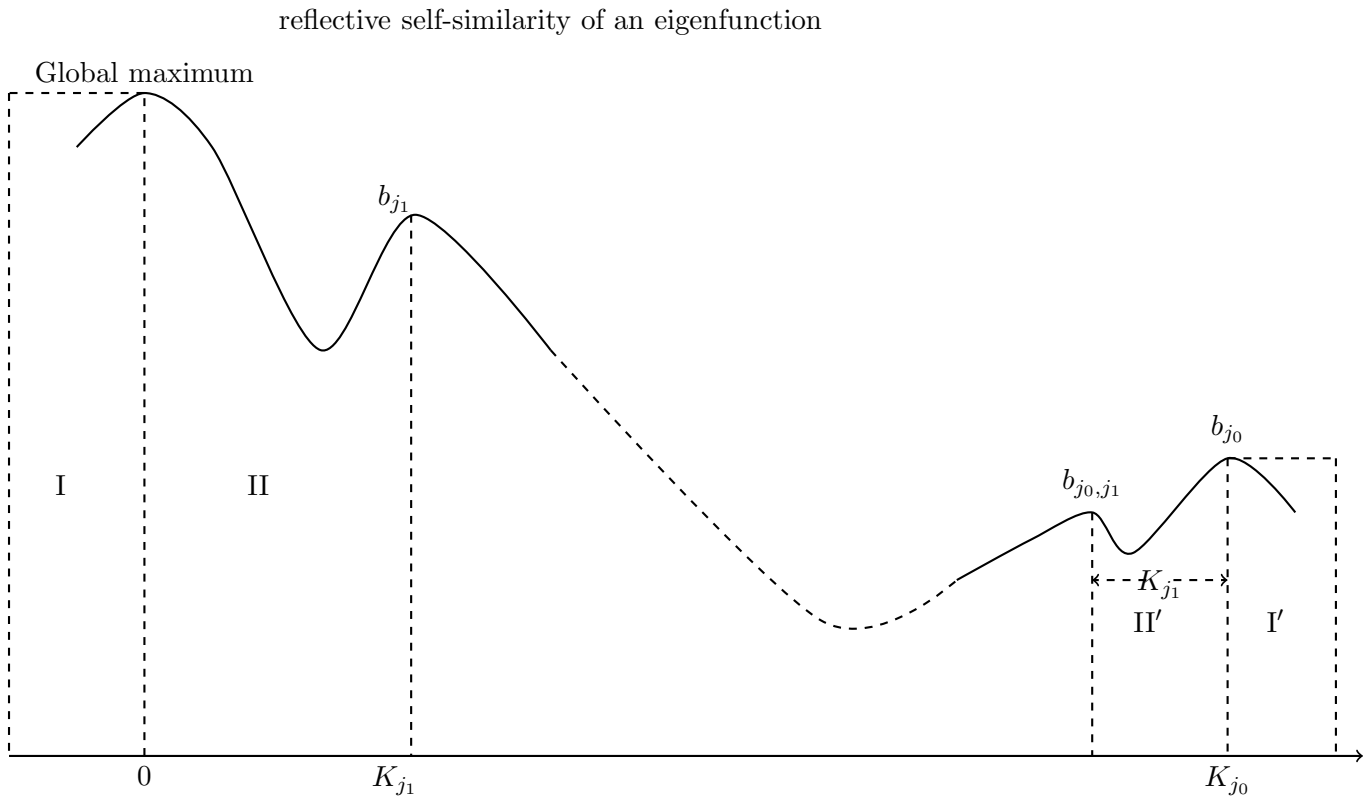
Fig.1: This depicts reflective self-similarity of an eigenfunction with global maximum at 0.
The self-similarity: I$^\prime$ is obtained from I  by scaling the $x$-axis  proportional to the ratio of the heights of the maxima in I and  I$^\prime$.
II$^\prime$ is obtained from  II  by scaling the $x$-axis
proportional to the ratio of the heights of  the maxima in II and
II$^\prime$. The behavior in the regions I$^\prime$, II$^\prime$
mirrors the behavior in regions I, II upon reflection and
corresponding dilation.
\vspace{4ex}

\par


 Our final main result is   the asymptotics of the transfer matrices.
Let $A_0=I$ and for $k\geq 1$,
\begin{equation*}
A_{k}(\theta)=\prod_{j=k-1}^{0 }A(\theta+j\alpha)=A(\theta+(k-1)\alpha)A(\theta+(k-2)\alpha)\cdots A(\theta)
\end{equation*}
and
\begin{equation*}
A_{-k}(\theta)=A_{k}^{-1}(\theta-k\alpha),
\end{equation*}
where $A(\theta)=\left(
             \begin{array}{cc}
               E-2\lambda\cos2\pi\theta & -1 \\
               1& 0\\
             \end{array}
           \right)
$.
$A_{k}$  is called the (k-step) transfer matrix. As is clear from the
definition, it also depends on $\theta$ and $E$ but since those
parameters will be usually fixed, we omit this from the notation.

We
define a new function $g :\mathbb{Z}\rightarrow \mathbb{R}^+$ as follows.

Case 1: If $x_0\cdot \ell \leq 0$ or $|x_0|>|\ell|$, set
\begin{equation*}
    g(\ell)=e^{|\ell|\ln|\lambda| }.
\end{equation*}

Case 2: If $x_0\cdot \ell \geq  0$ and $|x_0|\leq| \ell|\leq 2|x_0|$, set
\begin{equation*}
    g(\ell)=e^{(\ln\lambda-\eta)|\ell|}+e^{|2x_0-\ell|\ln|\lambda| }.
\end{equation*}

Case 3: If $x_0\cdot \ell \geq  0$ and $|\ell|> 2|x_0|$, set
\begin{equation*}
    g(\ell)=e^{(\ln\lambda-\eta)|\ell|}.
\end{equation*}
We have
\begin{theorem}\label{Thtransferasy}
Under the conditions of  Theorem \ref{Maintheoremdecay},   we have
\begin{equation}\label{Gtransferasy}
   g(\ell) e^{-\varepsilon|\ell|}\leq  ||A_{\ell}||\leq g(\ell) e^{\varepsilon|\ell|}.
\end{equation}

\end{theorem}
Let $\psi(\ell)$ denote  any solution to $
H_{\lambda,\alpha,\theta}\psi=E\psi$ that is linearly independent with $\phi(\ell)$. Let $\tilde{U} (\ell)=\left(\begin{array}{c}
        \psi(\ell)\\
       \psi(\ell-1)
     \end{array}\right)
 $. An immediate counterpart of (\ref{Gtransferasy}) is the following
\begin{corollary}\label{C.anysolution}
Under the conditions of Theorem \ref{Maintheoremdecay},  vectors  $\tilde{U}(\ell)$ satisfy
\begin{equation}\label{G.anysolution}
   g(\ell)e^{-\varepsilon|\ell|} \leq ||\tilde{U}(\ell)||\leq g(\ell)e^{\varepsilon|\ell|}.
\end{equation}
\end{corollary}
Our analysis also gives
\begin{corollary}\label{densityco}
                   Under the conditions of Theorem \ref{Maintheoremdecay},
                     we have,
\begin{enumerate}
\item [i)]
 $$ \limsup_{k\to \infty}\frac{\ln ||A_k||}{k}=\limsup_{k\to \infty}\frac{-\ln||U(k)||}{k}=\ln|\lambda|,$$

\item[ii)]
 $$ \liminf_{k\to \infty}\frac{\ln ||A_k||}{k}=\liminf_{k\to \infty}\frac{-\ln||U(k)||}{k}=\ln|\lambda|-\delta.$$
\item[iii)] outside a  sequence of lower density $1/2$,
                      \begin{equation}\label{densityE}
                       \lim_{k\to \infty}\frac{-\ln||U(k)||}{|k|}=\ln|\lambda|,
                      \end{equation}
                    \item[iv)]
                         outside a sequence of  lower density $0$,
                      \begin{equation}\label{densityT}
                       \lim_{k\to \infty}\frac{\ln ||A_k||}{|k|}=\ln|\lambda|.
                      \end{equation}
\end{enumerate}
\end{corollary}

Thus our analysis presents the second, after \cite{jl1}, study of the
dynamics of Lyapunov-Perron non-regular points, in a natural
setting. It is interesting to remark that (\ref{densityT}) also holds
throughout the pure point regime of \cite{jl1}. As in \cite{jl1}, the
fact that $g$ is not always the reciprocal of $f$ leads to exponential tangencies between
contracted and expanded directions with the rate
approaching $-\delta$ along a subsequence. Tangencies are an attribute of
 nonuniform hyperbolicity and are usually viewed as a difficulty to avoid
 through e.g. the parameter exclusion (e.g. \cite{bc,lsy,bj}). Our
 analysis allows to study them in detail and uncovers the hierarchical
structure of exponential
 tangencies positioned precisely at resonances. This will be explored
 in the future work. Finally we mention that the methods developed in  this paper have made it
 possible to determine the {\it exact} exponent of the exponential decay rate in
expectation for the two-point function \cite{jkl}, the first result of
this kind for any model.

The rest of this paper is organized  as follows. We list the
definitions and standard preliminaries in Section \ref{pre}.  Section \ref{loc}
is devoted to the  upper bound on the generalized eigenfunction in Theorem \ref{Maintheoremdecay}, establishing sharp
upper bounds for any eigensolution in the resonant case. This part of
the proof  has two technical ingredients similar to the arguments
used previously to prove localization in presence of
exponential frequency resonances. We present a universal version
of these statements in Theorem \ref{universalth} (a uniformity
statement for any Diophantine $\alpha$) and Theorem \ref{blockth} (a resonant block
expansion theorem for any one-dimensional operator), proved
correspondingly in   Appendices  A and B. Those statements can be of
use for proving localization for other models.  The rest of the argument is
based on new ideas specific to the phase resonance situation.
In Section \ref{js94} we prove the  sharp transition - Theorem \ref{Maintheorem}, and
lower bound on the generalized eigenfunctions in Theorem
\ref{Maintheoremdecay}. The part on the singular continuous spectrum, in particular, requires a new approach
to the palindromic argument  in order to remove a factor of four
inherent in the previous proofs, and the sharp lower bound in the
localization regime requires an even more delicate approach.
 In Section \ref{refl}, we use the local version of Theorem
 \ref{Maintheoremdecay} and establish reflective hierarchical structure of {\it
   resonances}  to prove the
reflective hierarchical structure Theorem       \ref{Universalend}. In Section \ref{transfer}, we study the growth of
transfer matrices  and  prove Theorem
\ref{Thtransferasy}, and Corollaries \ref{C.anysolution} and
\ref{densityco}. Except for the (mostly standard) statements listed in the preliminaries
and Lemma \ref{a1} this paper is entirely self-contained.

 \section{Preliminaries}\label{pre}
Without loss of generality, we assume $\lambda>1$ and $\ell>0$.
If $2\theta\in \alpha\Z+\Z$, then $\delta(\alpha,\theta)=0$, in which case the result  follows from \cite{j05}. Thus
in what follows we always assume        $2\theta\notin \alpha\Z+\Z$.

For any solution  of $H_{\lambda,\alpha,\theta}\varphi= E\varphi$,
 we have for   any $k,m$,
 \begin{equation}\label{G.new17}
    \left(\begin{array}{c}
                                                                                            \varphi(k+m) \\
                                                                                           \varphi(k+m-1)                                                                                        \end{array}\right)
                                                                                           =A_{k}(\theta+m\alpha)
 \left(\begin{array}{c}
                                                                                            \varphi(m) \\
                                                                                          \varphi(m-1)                                                                                        \end{array}\right).
 \end{equation}
 The Lyapunov exponent 
is given  by
 \begin{equation}\label{G21}
    L(E)=\lim_{k\rightarrow\infty} \frac{1}{k}\int_{\mathbb{R}/\mathbb{Z}} \ln \| A_k(\theta)\|d\theta.
 \end{equation}
The Lyapunov exponent can be computed precisely for $E$ in the
spectrum of $H_{\lambda,\alpha,\theta}$. We denote the spectrum by
$\Sigma_{\lambda,\alpha}$ (it  does not depend on $\theta$).
 \begin{lemma}\cite{MR1933451}\label{lya}
For $E\in \Sigma_{\lambda,\alpha}$ and $\lambda>1$, we have
 $L(E)=\ln\lambda$.
 \end{lemma}
 Recall that we always assume $E\in \Sigma_{\lambda,\alpha}$, so by
 upper semicontinuity and unique ergodicity,
one has
\begin{equation}\label{G23}
   \ln\lambda=\lim_{k\rightarrow\infty} \sup_{\theta\in\mathbb{R}/ \mathbb{Z}}\frac{1}{k} \ln \| A_k(\theta)\|,
\end{equation}
that is,  the convergence in (\ref{G23}) is  uniform   with respect to  $\theta\in\mathbb{R}$.
 Precisely, $ \forall \varepsilon >0$,
\begin{equation}\label{G24}
  \| A_k(\theta)\|\leq e^{(\ln\lambda+\varepsilon)k},  \text {for  }k  \text { large enough}.
\end{equation}


We start with the basic setup going back to
\cite{MR1740982}. Let us denote
$$ P_k(\theta)=\det(R_{[0,k-1]}(H_{\lambda,\alpha,\theta}-E) R_{[0,k-1]}).$$
It is easy to check  that
\begin{equation}\label{G34}
 A_{k}(\theta)=
\left(
  \begin{array}{cc}
   P_k(\theta) &- P_{k-1}(\theta+\alpha)\\
    P_{k-1}(\theta) & - P_{k-2}(\theta+\alpha) \\
  \end{array}
\right).
\end{equation}

    \par
    By Cramer's rule 
 for given  $x_1$ and $x_2=x_1+k-1$, with
     $ y\in I=[x_1,x_2] \subset \mathbb{Z}$,  one has
     \begin{eqnarray}
       |G_I(x_1,y)| &=&  \left| \frac{P_{x_2-y}(\theta+(y+1)\alpha)}{P_{k}(\theta+x_1\alpha)}\right|,\label{Cramer1}\\
       |G_I(y,x_2)| &=&\left|\frac{P_{y-x_1}(\theta+x_1\alpha)}{P_{k}(\theta+x_1\alpha)} \right|.\label{Cramer2}
     \end{eqnarray}
By  (\ref{G24})  and (\ref{G34}), the numerators in  (\ref{Cramer1}) and (\ref{Cramer2}) can be bounded uniformly with respect to $\theta$. Namely,
for any $\varepsilon>0$,
\begin{equation}\label{Numerator}
    | P_k(\theta)|\leq e^{(\ln \lambda+\varepsilon)k}
\end{equation}
for $k$ large enough.
\begin{definition}\label{Def.Regular}
Fix $\tau > 0$. A point $y\in\mathbb{Z}$ will be called $(\tau,k)$ regular  if there exists an
interval $[x_1,x_2]$  containing $y$, where $x_2=x_1+k-1$, such that
\begin{equation*}
  | G_{[x_1,x_2]}(y,x_i)|<e^{-\tau|y-x_i|} \text{ and } |y-x_i|\geq \frac{1}{40} k \text{ for }i=1,2.
\end{equation*}
\end{definition}

It is  easy to check that for any solution  of $H_{\lambda,\alpha,\theta}\varphi= E\varphi$, 
 \begin{equation}\label{Block}
   \varphi(x)= -G_{[x_1 ,x_2]}(x_1,x ) \varphi(x_1-1)-G_{[x_1 ,x_2]}(x,x_2) \varphi(x_2+1),
 \end{equation}
 where  $ x\in I=[x_1,x_2] \subset \mathbb{Z}$.


       \begin{definition}
     We  say that the set $\{\theta_1, \cdots ,\theta_{k+1}\}$ is $ \epsilon$-uniform if
      \begin{equation}\label{Def.Uniform}
        \max_{ x\in[-1,1]}\max_{i=1,\cdots,k+1}\prod_{ j=1 , j\neq i }^{k+1}\frac{|x-\cos2\pi\theta_j|}
        {|\cos2\pi\theta_i-\cos2\pi\theta_j|}<e^{k\epsilon}.
      \end{equation}
     \end{definition}
      Let $A_{k,r}=\{\theta\in\mathbb{R} \;|\;P_k(\cos2\pi  ( \theta -\frac{1}{2}(k-1)\alpha )  )|\leq e^{(k+1)r}\} $ with $k\in \mathbb{N}$ and $r>0$.
     We have the following Lemma.
      \begin{lemma}\label{Le.Uniform}(\text{Lemma 9.3 },\cite{MR2521117})
      Suppose  $\{\theta_1, \cdots ,\theta_{k+1}\}$ is  $ \epsilon_1$-uniform. Then there exists some $\theta_i$ in set  $\{\theta_1, \cdots ,\theta_{k+1}\}$ such that
     $\theta_i\notin A_{k,\ln\lambda-\epsilon}$  if   $ \epsilon>\epsilon_1$ and $ k$
      is sufficiently large.
      \end{lemma}

\section{Localization}\label{loc}

Let $\alpha$  be Diophantine and $\delta(\alpha,\theta)$ be given by (\ref{G.delta}).
                    Suppose $\ln |\lambda| >\delta(\alpha,\theta)$.
                    Recalling that for  $E$ a generalized eigenvalue of $H_{\lambda,\alpha,\theta}$ and  $\phi$
                     the corresponding generalized eigenfunction,
                    we denote  $ U(\ell) =\left(\begin{array}{c}
                             \phi(\ell)\\
                            \phi({\ell-1})
                          \end{array}\right)
                      $.
In this part we will prove the localization part  of Theorem \ref{Maintheorem} and the upper bound of Theorem \ref{Maintheoremdecay}.
That is
 \begin{theorem}\label{MaintheoremAL}

                    For any $\varepsilon>0$, there exists $K$ such that for any $|\ell|\geq K$,  $U(\ell)$ satisfies
                     \begin{equation}\label{G.upper}
                    ||U(\ell)||\leq f(\ell)e^{\varepsilon|\ell|}.
                      \end{equation}
                            In particular, $H_{\lambda,\alpha,\theta}$  satisfies Anderson localization,
                            and  the following upper bound holds for the generalized eigenfunction,
                            \begin{equation}\label{decayingupper}
                                 ||U(\ell)||\leq e^{-(\ln\lambda-\delta-\varepsilon)|\ell|}.
                            \end{equation}
                     \end{theorem}

  By  Schnol's Theorem \cite{berezanskii1968expansions} if every generalized eigenfunction of $H$ decays
  exponentially,  then $H$  satisfies Anderson  localization.  Thus,
  by Remark \ref{rem}.2, in order to prove Theorem \ref{MaintheoremAL}, it suffices to prove the first
  part of  Theorem \ref{MaintheoremAL}.
  \par

  Without loss of generality assume $|\phi(0)|^2+|\phi(-1)|^2=1$. Let
  $ \psi$ be  any  solution of $H_{\lambda,\alpha,\theta} \psi=E\psi$ linear independent with  $\phi$, i.e.,
  $|\psi(0)|^2+|\psi(-1)|^2=1$ and
  \begin{equation*}
    \phi(-1)\psi(0)-\phi(0)\psi(-1)=c,
  \end{equation*}
  where $c\neq 0$.

  Then
    by the constancy of the Wronskian, one has
  \begin{equation}\label{W}
    \phi(y+1)\psi(y)-\phi(y)\psi(y+1)=c.
  \end{equation}
  We also will denote by $\varphi$  an {\it arbitrary} solution, so
  either $\psi$ or $\phi$, and denote
by $ U^{\varphi}(y) =\left(\begin{array}{c}
                             \varphi(y)\\
                            \varphi({y-1})
                          \end{array}\right)
                      $.
  Let $ U(y) =\left(\begin{array}{c}
                             \phi(y)\\
                            \phi({y-1})
                          \end{array}\right)
                      $
 and
 $ \tilde{U}(y) =\left(\begin{array}{c}
                             \psi(y)\\
                            \psi({y-1})
                          \end{array}\right)
                      $.

                      Below $\varepsilon>0$ is always  sufficiently small and  $\frac{p_n}{q_n}$ is the continued fraction expansion of $\alpha$.


We will make a repeated use of the following two Theorems that can be
useful also in other situations. The first
theorem is an arithmetic statement that holds for any Diophantine $\alpha$.
\begin{theorem}(Uniformity Theorem)\label{universalth}

  Let $I_1,I_2$ be two disjoint intervals in $\Z$ such that $\# I_1= s_1 q_n$ and $\#I_2= s_2q_n$, where $s_1,s_2\in \Z^+$.
  Suppose $|j|\leq C sq_n$ for any $j\in I_1\cup I_2$ and $|s|\leq q_n^C$, where $s=s_1+s_2$. Let $\gamma>0$ be such that
  \begin{equation}\label{appG1}
    e^{-\gamma sq_n}=\min_{i,j\in I_1\cup I_2}|\sin\pi (2\theta+(i+j)\alpha)|.
  \end{equation}
  Then for any $\varepsilon>0$, $\{\theta_{j}=\theta +j\alpha\}_{j\in
    I_1\cup I_2}$ is $\gamma+\varepsilon$ uniform if $n$ is large
  enough (not depending on $\gamma$).
\end{theorem}

The second theorem holds for any one-dimensional (not necessarily
quasiperiodic or even ergodic) Schr\"odinger operator. It is the
technique to establish exponential decay with respect to the distance to the resonances .
\begin{theorem}(Block Expansion Theorem)\label{blockth}

Fix $\gamma>0$.
Let $r_{y}^{\varphi}=\max_{|\sigma|\leq 10 \gamma}|\varphi(y+\sigma k)|$.
Suppose    $y_1,y_2\in \Z$ are such that $y_2-y_1=   k$.
 Suppose  there exists some $\tau>0$ such that for any  $ y\in [y_1+\gamma k, y_2-\gamma k]$,
   $y$ is $(\tau,k_1 )$ regular,  for some $\frac{\gamma}{20} k<k_1\leq \frac{1}{2} \min\{|y-y_1|, |y-y_2|\}$.
  Then  for large enough $k$,
 \begin{equation}\label{appfirst}
  r^\varphi_y \leq\max\{ r_{y_1}^{\varphi}  \exp\{-\tau(|y-y_1|-3\gamma k)\},
  r_{y_2}^{\varphi}\exp\{-\tau(|y-y_2|-3\gamma k )\}\},
 \end{equation}
 for all $y\in[y_1+10\gamma k,y_2-10\gamma k]$.
\end{theorem}
These two theorems  are  similar in spirit to the statements
in \cite{jl1} with the ones related to Theorem \ref{universalth} being  in turn   modifications of the ones appearing in
\cite{MR2521117,MR3340177,MR3292353}.
While these techniques were
developed specifically to treat the non-Diophantine case, these ideas
turn out to be relevant for the case of phase resonances as well.  Theorem \ref{blockth} is essentially the block-expansion technique of multiscale analysis, e.g. \cite{fsw}, coupled with certain extremality arguments, an idea used also in \cite{jl1}. We expect Theorem \ref{universalth} to be useful for various one-frequency quasiperiodic problems, and Theorem \ref{blockth} for general one-dimensional models. We
present the proofs  in Appendix A and  B respectively.

The following Lemma establishes the non-resonant decay.
\begin{lemma}\label{Keylemma}
 Suppose  $k_0\in[-2Ck,2Ck]$ is such that
 \begin{equation*}
  |\sin\pi(2\theta+\alpha  k_0)|=\min_{|x|\leq 2 Ck}
  |\sin\pi(2\theta+\alpha  x)|,
 \end{equation*}
 where $C\geq1$ is a   constant.
 Let $\gamma,\varepsilon$ be small positive constants.
Let $y_1=0, y_2=k_0, y_3=y^{\prime}$.
Assume $y$ lies in  $[y_i,y_j]$ (i.e., $y\in [y_i,y_j]$)with  $|y_i-y_j|\geq k$.
 Suppose     $|y_i|,|y_j|\leq Ck$ and $|y-y_i|\geq 10\gamma k$, $|y-y_j|\geq 10\gamma k$.
  Then for large enough $k$,
 \begin{equation}\label{first}
  r^\varphi_y \leq\max\{ r_{y_i}^{\varphi}  \exp\{-(\ln \lambda- \varepsilon)(|y-y_i|-3\gamma k)\},
  r_{y_j}^{\varphi}\exp\{-(\ln \lambda- \varepsilon)(|y-y_j|-3\gamma k )\}\}.
 \end{equation}
\end{lemma}
\begin{proof}
 By the DC on $\alpha$,  there exist $\tau^\prime,\kappa^\prime>0$ such that for any $x\neq k_0$ and $|x|\leq 2Ck$,
\begin{equation}\label{smalldivisorcondition2new}
   |\sin\pi(2\theta+x\alpha)|\geq \frac{\tau^\prime}{k^{\kappa^\prime}}.
 \end{equation}
 Fix $y^\prime$.
For any $p$ satisfying  $|p-y^\prime|\geq \gamma k$, $|p|\geq \gamma k$ and $|p-k_0|\geq \gamma k$, let
\begin{equation*}
    d_p=\frac{1}{10}\min\{|p|,|p-k_0|,|p-y^\prime|\}
\end{equation*}
Let $\frac{p_n}{q_n}$ be the continued fraction expansion of $\alpha$.
Let $n$ be the largest integer such that
\begin{equation*}
   2 q_n \leq  d_p,
\end{equation*}
and  let $s$ be the largest positive integer such that $2sq_n\leq
 d_p$. Notice that $ 2 q_{n+1}> d_p $ and by the Diophantine condition on $\alpha$, we have $s\leq q_n^C$.

Case 1: $0\leq k_0< p$.
 We construct intervals
\begin{equation*}
    I_1=[-2s q_n,-1],I_2=[p- 2sq_n,p+ 2sq_n-1].
\end{equation*}
Case 2: $0\leq  p< k_0$.

If $p\leq \frac{k_0}{2}$,
we construct intervals
\begin{equation*}
    I_1=[-2sq_n,2s q_n-1],I_2=[p- 2sq_n,p-1].
\end{equation*}

If $p> \frac{k_0}{2}$, we construct intervals
\begin{equation*}
    I_1=[- 2sq_n, 2sq_n-1],I_2=[p,p+2sq_n-1].
\end{equation*}
Case 3: $p< k_0\leq 0$.

 We construct intervals
\begin{equation*}
    I_1=[0,2sq_n-1],I_2=[p-2sq_n ,p+2sq_n-1].
\end{equation*}
Case 4: $k_0<  p < 0$.

If $p\leq \frac{k_0}{2}$,
we construct intervals
\begin{equation*}
    I_1=[- 2sq_n,2sq_n-1],I_2=[p- 2sq_n,p-1].
\end{equation*}

If $p> \frac{k_0}{2}$, we construct intervals
\begin{equation*}
    I_1=[-2sq_n,2sq_n-1],I_2=[p,p+2sq_n-1].
\end{equation*}
Case 5: $k_0\leq 0< p$.
\begin{equation*}
    I_1=[0,2sq_n-1],I_2=[p-2sq_n,p+2sq_n-1].
\end{equation*}
Case 6: $p< 0\leq k_0$.
\begin{equation*}
    I_1=[-2sq_n,-1],I_2=[p-2sq_n,p+2sq_n-1].
\end{equation*}
Using the small divisor condition (\ref{smalldivisorcondition2new}) and the construction of $I_1, I_2$, we have in any case
\begin{equation*}
    \min_{i,j\in I_1\cup I_2}|\sin\pi (2\theta+(i+j)\alpha)|\geq \frac{\tau^\prime}{k^{\kappa^\prime}}.
\end{equation*}

By Theorem  \ref{universalth},   for any $\varepsilon>0$, we have,  in each case $\{\theta_{j}=\theta +j\alpha\}_{j\in I_1\cup I_2}$ is $\varepsilon$ uniform.
Combining with  Lemma  \ref{Le.Uniform}, there exists some $j_0$ with  $j_0\in I_1\cup I_2$
    such that
      $ \theta_{j _0}\notin  A_{6sq_n-1,\ln\lambda-\varepsilon }$.

We have the following simple Lemma that will be used repeatedly in the
rest of the paper.

\begin{lemma}\label{lem}
Let $a_n\to\infty$, $0<t<1.$ Then for sufficiently large $n$ and
$|j|<t a_n$ we have $\theta_j=\theta+j\alpha \in
A_{2a_n-1,\ln\lambda-\varepsilon}.$
\end{lemma}

{\bf Proof:}
 Assume  for some $|j|<t a_n,$  $\theta_j\notin
A_{2a_n-1,\ln\lambda-\varepsilon}.$

     Let $I=[j-a_n+1,j+a_n-1]=[x_1,x_2]$. We have $x_1<0<x_2$ and
\begin{equation}\label{xi}|x_i|>(1-t)a_n.\end{equation} By    (\ref{Cramer1}), (\ref{Cramer2}) and (\ref{Numerator}),
 one has
\begin{equation*}
|G_I(0,x_i)|\leq e^{(\ln\lambda+\varepsilon )(2a_n-1-|x_i|)-(2a_n-1)(\ln\lambda -\varepsilon)}.
\end{equation*}
Using (\ref{Block}), we obtain
\begin{equation}\label{0singular}
  |\phi(-1)|,  |\phi(0)|  \leq \sum_{i=1,2} e^{\varepsilon a_n}|\varphi(x_i^{\prime})|e^{-|x_i|\ln \lambda },
\end{equation}
where $x_1^{\prime}=x_1-1 $ and $x_2^{\prime}=x_2+1 $. Because of (\ref{xi}),
 (\ref{0singular}) implies  $ |\phi(-1)|,  |\phi(0)|\leq e^{-(1-t-\varepsilon)\ln \lambda a_n}$.
This  contradicts   $|\phi(-1)|^2+  |\phi(0)|^2=1 $.\qed

Lemma \ref{lem} implies that $j_0$ must belong to $I_2$.

Set $I=[j_0-3sq_n+1,j_0+3sq_n-1]=[x_1,x_2]$.   By    (\ref{Cramer1}), (\ref{Cramer2}) and (\ref{Numerator}) again,
 one has
\begin{equation*}
|G_I(p,x_i)|\leq e^{(\ln\lambda+\varepsilon )(6sq_n-1-|k-x_i|)-(6sq_n-1)(\ln\lambda -\varepsilon)}\leq e^{ \varepsilon sq_n}e^{-|p-x_i|\ln \lambda }.
\end{equation*}
 Notice that $|p-x_1|,|p-x_2|\geq sq_n-1$. Thus
 for any  $p\in[y_i+\gamma k, y_j-\gamma k]$,  $p$ is $(6sq_n-1,\ln\lambda-\varepsilon)$ regular. Block expansion   (Theorem \ref{blockth}) now implies the  Lemma.
\end{proof}
\begin{remark}
Recall that  $ U^{\varphi}(y) =\left(\begin{array}{c}
                             \varphi(y)\\
                            \varphi({y-1})
                          \end{array}\right)
                      $.
                      By (\ref{G.new17}) and (\ref{G24}), we have
  \begin{equation}\label{G.new18}
Ce^{-(\ln\lambda+\varepsilon)|k_1-k_2|}||U^{\varphi}(k_2)|| \leq  ||U^{\varphi}(k_1)||\leq  Ce^{(\ln\lambda+\varepsilon)|k_1-k_2|}||U^{\varphi}(k_2)||.
\end{equation}
Thus  (\ref{first}) implies
\begin{equation}\label{second}
   ||U^{\varphi} (y)||\leq \max\{||U^{\varphi} (y_i)||\exp\{-(\ln \lambda- \varepsilon)(|y-y_i|-14\gamma k)\},
   ||U^{\varphi} (y_j)||\exp\{-(\ln \lambda- \varepsilon)(|y-y_j|-14\gamma k)\}\}
\end{equation}
\end{remark}
\begin{lemma}\label{Keylemma1}
Fix  $0<t<\ln\lambda$. 
Suppose
 \begin{equation}\label{ksmallG}
  |\sin\pi(2\theta+\alpha  k)|= e^{-t |k|}.
 \end{equation}
Then  for  large  $|k|$
 \begin{equation}\label{Equ9}
 ||U^{\varphi}(k)||\leq \max\{ ||U^{\varphi}(0)||, ||U^{\varphi}(2k)|| \} e^{-(\ln\lambda -t-\varepsilon) |k|}.
 \end{equation}
\end{lemma}
\begin{proof}

Without loss of generality assume $k>0$.
By the DC on $\alpha$,
we have
\begin{equation*}
    |\sin\pi(2\theta+\alpha  k)|=\min_{|x|\leq 8k}|\sin\pi(2\theta+\alpha  x)|.
\end{equation*}
Furthermore,
there exist $\tau^\prime,\kappa^\prime>0$ such that for any $x\neq k$ and $|x|\leq 8k$,
\begin{equation*}
   |\sin\pi(2\theta+x\alpha)|\geq \frac{\tau^\prime}{k^{\kappa^\prime}}.
 \end{equation*}
Let $\gamma$ be  any small positive constant and define  $r_{y}^{\varphi}=\max_{|\sigma|\leq 10\gamma}|\varphi(y+\sigma k)|$ .
Let $\frac{p_n}{q_n}$ be the continued fraction expansion of $\alpha$.
Let $n$ be the largest integer such that
\begin{equation*}
   ( \frac{t+C\varepsilon }{\ln\lambda-t-C\varepsilon}+1) q_n \leq \frac{k}{2},
\end{equation*}
where  $C$ is a large constant depending on $\lambda,t$.
Let $s$ be the largest positive integer such that $sq_n\leq
 \frac{1}{2}k$.
 Then $s> \frac{t+C\varepsilon}{\ln\lambda-t-C\varepsilon}$.  Combining with  the fact
 $(s+1)q_n\geq \frac{1}{2}k $,
 we obtain
 \begin{equation}\label{a}
    2s\frac{q_n}{k}> \frac{t}{\ln\lambda}+C\varepsilon.
 \end{equation}

Construct intervals
\begin{equation*}
    I_1=[-sq_n,sq_n-1], I_2=[k-sq_n,k+sq_n-1].
\end{equation*}

Let $\theta_j=\theta+j\alpha$ for $j\in I_1\cup I_2$. The set $\{\theta_j\}_{j\in I_1\cup I_2}$
consists of $4sq_n$ elements.

   By Theorem \ref{universalth} and (\ref{ksmallG}), we have    $\{\theta_j\}_{j\in I_1\cup I_2}$ is $\frac{t k}{4sq_n}+\varepsilon$ uniform.
 Combining with  Lemma  \ref{Le.Uniform}, there exists some $j_0$ with  $j_0\in I_1\cup I_2$
    such that
      $ \theta_{j _0}\notin  A_{4sq_n-1,\ln\lambda-\frac{t k}{4sq_n}-\varepsilon }$.

      First assume $j_0\in I_2$.

      Set $I=[j_0-2sq_n+1,j_0+2sq_n-1]=[x_1,x_2]$.   By    (\ref{Cramer1}), (\ref{Cramer2}) and (\ref{Numerator}),
 one has
\begin{equation*}
|G_I(k,x_i)|\leq e^{(\ln\lambda+\varepsilon )(4sq_n-1-|k-x_i|)-(4sq_n-1)(\ln\lambda -\frac{t k}{4sq_n}-\varepsilon)}.
\end{equation*}
Using (\ref{Block}), we obtain
\begin{equation}\label{IIIcase1}
  |\varphi(k-1)|,  |\varphi(k)|  \leq \sum_{i=1,2} e^{(t+\varepsilon) k}|\varphi(x_i^{\prime})|e^{-|k-x_i|\ln \lambda },
\end{equation}
where $x_1^{\prime}=x_1-1 $ and $x_2^{\prime}=x_2+1 $.

Fix small $\gamma=\frac{\varepsilon}{C}$, where $C$ is a large constant depending on $\lambda,t$.

If  $x_i^{\prime}\in[-10\gamma k, 10\gamma k]$, $x_i^{\prime}\in[k-10\gamma k,k+10\gamma k]$ or
$x_i^{\prime}\in[2k-10\gamma k,2k+10\gamma k]$,  we bound  $\varphi(x_i^{\prime})$ in (\ref{IIIcase1}) by
 $r_0^{\varphi}$, $r_{k}^{\varphi}$ or $r_{2k}^{\varphi}$ respectively.
In other cases, we  bound  $\varphi(x_i^{\prime})$ in (\ref{IIIcase1}) with (\ref{first}) using $k_0=k,y=x_i^{\prime}$ and $y^\prime=-k,2k$ or $3k$. Then
we have
\begin{eqnarray*}
  |\varphi(k-1)|,  |\varphi(k)| &\leq& \max\{r_{k\pm 2k}^{\varphi} \exp\{-(2\ln \lambda - t- C\gamma-\varepsilon)k\},
  r_{k\pm k}^{\varphi} \exp\{-(\ln \lambda - t- C\gamma-\varepsilon)k\}   ,\\
   && r_{k} ^{\varphi} \exp\{-(\ln \lambda -
   \varepsilon)2sq_n+(t+C\gamma)k\}\}.
\end{eqnarray*}

However by (\ref{a}), we have that
\begin{eqnarray*}
  | \varphi(k-1)|,  | \varphi(k)| &\leq& r_{k}^{\varphi}  \exp\{-(\ln \lambda - \varepsilon)2sq_n+(t+C\gamma)k\} \\
   &\leq &  e^{-\varepsilon k}r_{k}^{\varphi}
\end{eqnarray*}
can not happen, so we must have
\begin{equation}\label{G.Secondadd1new}
  |\varphi(k-1)|,  |\varphi(k)|\leq  \exp\{-(\ln \lambda - t-C\gamma- \varepsilon)k\}\max\{ r_{k\pm k}^{\varphi},  e^{-k\ln\lambda} r_{k\pm 2k}^{\varphi} \}.
\end{equation}
Notice that by (\ref{G.new18}), one has
\begin{equation*}
     r_{k\pm 2k}^{\varphi}\leq  e^{(\ln\lambda+C\gamma)k} r_{k\pm k}^{\varphi} .
\end{equation*}
Then (\ref{G.Secondadd1new}) becomes
\begin{equation*}
  ||U^{\varphi}(k)|| \leq  \exp\{-(\ln \lambda - t-C\gamma- \varepsilon)k\}\max\{ r_{0}^{\varphi},   r_{2k}^{\varphi} \}.
\end{equation*}
By (\ref{G.new18}) again, one has
\begin{equation*}
r^{\varphi}_{y}e^{-(\ln\lambda+\varepsilon)10\gamma k}\leq    ||U^{\varphi}(y)||\leq  r^{\varphi}_{y}e^{(\ln\lambda+\varepsilon)10\gamma k}.
\end{equation*}
Thus
\begin{eqnarray}
  ||U^{\varphi}(k)|| &\leq & \max\{ ||U^{\varphi}(0)||, ||U^{\varphi}(2k)|| \} e^{-(\ln\lambda -t-C\gamma-\varepsilon) |k|} \nonumber \\
   &\leq &  \max\{ ||U^{\varphi}(0)||, ||U^{\varphi}(2k)|| \} e^{-(\ln\lambda -t-\varepsilon) |k|}\label{Equ6}.
\end{eqnarray}
This implies (\ref{Equ9}). Thus in order to prove the lemma, it suffices to exclude the case $j_0\in I_1$.

  Suppose  $j_0\in I_1$. Notice that $I_1+k=I_2$(i.e., $I_2$ can be obtained from $I_1$ by moving by $k$ units).
  Following the proof of (\ref{Equ6}), we get (move $-k$ units in (\ref{Equ6}))
  \begin{equation*}
    ||U^{\varphi}(0)||\leq \max\{ ||U^{\varphi}(-k)||, ||U^{\varphi}(k)|| \} e^{-(\ln\lambda -t-\varepsilon) |k|}.
  \end{equation*}
    This  contradicts  $||U^{\phi}(0)||=1$.
\end{proof}
{\bf Proof of Theorem \ref{MaintheoremAL}}
Without loss of generality, assume $\ell>0$.

For  any $\varepsilon>0$,  let  $\gamma=\frac{\varepsilon}{C}>0$, where $C$ is a large constant that may depend  on $\lambda$ and $\delta$.
Let  $x_0^\prime$ (we can choose any one if $x_0^\prime$ is not unique) be such that
              \begin{equation*}
             |\sin\pi(2\theta+x_0^\prime \alpha)|  = \min_{|x|\leq 4|\ell|}|\sin\pi(2\theta+x\alpha)|.
              \end{equation*}
             Let $\eta^{\prime}\in (0, \infty)$ be given by the following equation,
                \begin{equation}\label{etaprime}
                   |\sin\pi(2\theta+x_0^\prime\alpha)|=e^{-\eta ^\prime|\ell|}.
                \end{equation}

Case 1: $|\sin\pi(2\theta+x_0^\prime\alpha)|\neq |\sin\pi(2\theta+x_0\alpha)|$. This implies $|x_0^\prime|>2\ell$. In this case for any $\varepsilon>0$,
we have
 $ \eta\leq \varepsilon$  if $\ell$ is large enough  by the Diophantine condition.
  Let $y=\ell$,  $C=2$, $k=2\ell$, and $y^\prime=2\ell$ in Lemma \ref{Keylemma}.  Then  $k_0=x_0^\prime $ and we obtain
  \begin{equation*}
    |\phi(\ell)|,|\phi(\ell-1)|\leq e^{-(\ln\lambda-C\gamma)\ell},
  \end{equation*}
  This implies the right inequality of (\ref{G.Asymptotics}) in this case.

Case 2: $|\sin\pi(2\theta+x_0^\prime\alpha)|= |\sin\pi(2\theta+x_0\alpha)|$, so $\eta=\eta^\prime$.

If $x_0\leq 0$,  let $y=\ell$,  $C=2$, $k=2\ell$ and $y^\prime=2\ell$ in Lemma \ref{Keylemma}. Then Theorem \ref{MaintheoremAL} holds by (\ref{first}).

Now we consider the case $x_0> 0$.

We split the proof into two subcases.

Subcase i: $\eta\leq\gamma$.

 Fix some $y\in[\gamma\ell,2\ell- \gamma\ell]$.
 Let $n$ be such that $q_n\leq \frac{1}{20}\min\{y,2\ell-y\}<q_{n+1}$, and let $s$ be the largest positive integer such that $sq_n\leq \frac{1}{20}\min\{y,2\ell-y\}$.
  We construct intervals
  \begin{equation*}
    I_1=[-2sq_n,2sq_n-1], I_2=[y-2sq_n,y-1].
  \end{equation*}
  By the definition of $\eta^\prime,\eta$ and construction of $I_1,I_2$, we have
  \begin{equation*}
    \min_{i,j\in I_1\cup I_2}|\sin\pi (2\theta+(j+i)\alpha)| \geq e^{-\eta^{\prime}\ell}=e^{-\eta\ell}.
  \end{equation*}
  By Theorem \ref{universalth},  we get  $\{\theta_{j}=\theta +j\alpha\}_{j\in I_1\cup I_2}$ is $2\gamma$ uniform.
   As in the  proof of Lemma \ref{Keylemma}, we obtain there exists some $j_0$ with  $j_0\in  I_2$
    such that
      $ \theta_{j _0}\notin  A_{6sq_n-1,\ln\lambda- 3\gamma }$. Thus
  $y$ is $(\ln\lambda-3\gamma,6sq_n)$ regular. 
  By block expansion  (Theorem \ref{blockth} with $y_1=0,y_2=2\ell,\tau=\ln\lambda-3\gamma$), we  get
  \begin{equation*}
    |\phi(\ell)|,|\phi(\ell-1)|\leq e^{-(\ln\lambda-C\gamma)\ell},
  \end{equation*}
  This implies the right inequality of (\ref{G.Asymptotics}).

   Subcase ii: $\eta\geq  \gamma$.

   By the definition of $\delta(\alpha,\theta)$ and the fact that $\delta(\alpha,\theta)< \ln\lambda$, we must have
   \begin{equation}\label{largex0}
 \frac{\gamma}{\ln\lambda} \ell  \leq  |x_0|\leq  2\ell.
   \end{equation}

Applying  Lemma \ref{Keylemma1} with $k=x_0$ to the generalized eigenfunction $\phi(k)$, we have
\begin{equation}\label{x0}
    ||U(x_0)||=||U^{\phi}(x_0)||\leq e^{-(\ln\lambda-\varepsilon)|x_0|} e^{\eta \ell}.
\end{equation}
Applying     Lemma \ref{Keylemma} with $y=\ell, k=2\ell, C=2,  y^\prime=2\ell, k_0=x_0,\varphi=\phi$, considering  $\ell>x_0$ and
 $\ell\leq x_0$ separately, and using (\ref{x0}), we obtain  Theorem \ref{MaintheoremAL}.
\begin{remark}\label{Rcase1}

 By (\ref{G24}), we have
  \begin{equation*}
   || U(\ell)||\geq ||A_{\ell}||^{-1}|| U(0)||\geq e^{-(\ln\lambda+\varepsilon)\ell}.
  \end{equation*}
   This already  implies the left inequality of (\ref{G.Asymptotics}), except for the   Subcase ii.
\end{remark}

\section{Palindromic arguments }\label{js94}
\subsection{Singular continuous spectrum}
 We first  show that if $\ln|\lambda|< \delta(\alpha,\theta)$, then
$H_{\lambda,\alpha,\theta}$ has purely singular continuous spectrum, which is the second part of Theorem \ref{Maintheorem}. That is
\begin{theorem}\label{Maintheoremsingular}
Let $H_{\lambda,\alpha,\theta}$ be an almost Mathieu operator with $|\lambda|>1$.
For   any irrational number $\alpha$ and $\theta\in \R$, define $\delta(\alpha,\theta)\in[0,\infty]$ by (\ref{G.delta}).
Then  $H_{\lambda,\alpha,\theta}$ has purely singular continuous spectrum  if $\ln|\lambda|< \delta(\alpha,\theta)$.
\end{theorem}
Actually, we can prove a more general result.
\begin{theorem}\label{general}
Let $H_{v,\alpha,\theta}$  be a discrete Schr\"{o}dinger operator,
 \begin{equation*}
 (H_{v,\alpha,\theta}u)(n)=u({n+1})+u({n-1})+  v(\theta+n\alpha)u(n), 
 \end{equation*}
 where $v:\mathbb{T}\longrightarrow \R$  is an even Lipchitz continuous function.
For   any irrational number $\alpha$ and $\theta\in \R$, define $\delta(\alpha,\theta)\in[0,\infty]$ by (\ref{G.delta}).
Then  $H_{v,\alpha,\theta}$ has  no  eigenvalues     in the  regime $\{E\in \R: L(E)<\delta(\alpha,\theta)\}$, where  $L(E)$ is the Lyapunov exponent.
\end{theorem}
Theorem \ref{Maintheoremsingular} follows directly from Theorem \ref{general}, Lemma \ref{lya} and Kotani theory.

By the definition of $\delta(\alpha,\theta)$, for any $\varepsilon>0$ there exists a sequence $\{k_i\}_{i=1}^{\infty}$
such that
\begin{equation}\label{G.appi}
    ||2\theta+k_i\alpha||_{\R/\Z}\leq e^{-(\delta-\varepsilon)|k_i|}.
\end{equation}
Without loss of generality assume $k_i>0 $.


{\bf Proof of Theorem \ref{general} }

\begin{proof}
Suppose not. Let $u$ be an   $\ell^2(\Z)$ solution, i.e.,
$H_{v,\alpha,\theta}u=Eu$, with  $L(E)< \delta(\alpha,\theta)$. Without loss of generality assume
\begin{equation*}
 || u||_{\ell^2} = \sum_n|u(n)|^2=1.
\end{equation*}

We let $u_i(n)= u(k_i-n)$,
 $V(n)=v(\theta+n\alpha)$ and $V_i(n)=v(\theta+(k_i-n)\alpha)$.
 Then by (\ref{G.appi}), evenness  and   Lipchitz continuity of $v$, one has for all $n\in\Z$,
 \begin{equation}\label{Eqp}
    |V(n)-V_i(n)|\leq Ce^{-(\delta-\varepsilon)|k_i|}.
 \end{equation}
 We also have
 \begin{equation}\label{Equ}
 u(n+1)+u(n-1)+ V(n)u(n)=Eu(n)
 \end{equation}
 and
  \begin{equation}\label{Equi}
 u_i(n+1)+u_i(n-1)+ V(n)u_i(n)=Eu_i(n).
 \end{equation}
 Let $W(n)=W(f,g)=f(n+1)g(n)-f(n)g(n+1)$ be the
Wronskian, as usual, and let

 \begin{equation*}
   \Phi(n)=\left(\begin{array}{cc}
                   u(n) \\ u(n-1)
                 \end{array}
   \right);
   \Phi_i(n)=\left(\begin{array}{cc}
                   u_i(n) \\ u_i(n-1)
                 \end{array}
   \right).
 \end{equation*}

By a standard calculation using (\ref{Eqp}), (\ref{Equ}) and (\ref{Equi}),
we have
\begin{eqnarray}
   |W(u,u_i)(n)-W(u,u_i)(n-1)| &\leq & |V(n)-V_i(n)||u(n)u_i(n)|  \nonumber \\
  &\leq &  C e^{-(\delta-\varepsilon)|k_i|}|u(n)u_i(n)|. \nonumber
\end{eqnarray}
This implies for any $m>0$ and $n$,
\begin{eqnarray}
   |W(u,u_i)(n+m)-W(u,u_i)(n-1)| &\leq & C e^{-(\delta-\varepsilon)|k_i|}\sum_{j=0}^{m-1} |u(n+j)u_i(n+j)|  \nonumber \\
  &\leq &  C e^{-(\delta-\varepsilon)|k_i|}, \label{Equ7}
\end{eqnarray}
where the second inequality holds by the fact $||u||_{\ell^2}=||u_i||_{\ell^2}=1.$

Notice that 
$\sum_n|W(u,u_i)(n)|\leq 2$ . Thus for some $n $, one has
\begin{equation*}
    |W(u,u_i)(n)|\leq C e^{-(\delta-\varepsilon)|k_i|}.
\end{equation*}
By (\ref{Equ7}), we must have that
\begin{equation}\label{Equ8}
    |W(u,u_i)(n)|\leq C e^{-(\delta-\varepsilon)|k_i|}
\end{equation}
holds for all $n$.

Now we split $k_i$ into odd or even to discuss the problem.

Case 1. $k_i$ is even. Let $m_i=\frac{k_i}{2}$,
then
 \begin{equation*}
   \Phi(m_i)=\left(\begin{array}{cc}
                   u(m_i) \\ u(m_i-1)
                 \end{array}
   \right);
   \Phi_i(m_i)=\left(\begin{array}{cc}
                   u( m_i) \\ u(m_i+1)
                 \end{array}
   \right).
 \end{equation*}

 Applying (\ref{Equ8}) with $n=m_i-1$, we have
 \begin{equation*}
    |u(m_i)||u(m_i+1)-u(m_i-1)|\leq C e^{-(\delta-\varepsilon)|k_i|}.
 \end{equation*}
 This implies
 \begin{equation}\label{Case11}
    |u(m_i)|\leq C e^{-\frac{1}{2}(\delta-\varepsilon)|k_i|},
 \end{equation}
 or
 \begin{equation}\label{Case12}
    |u(m_i+1)-u(m_i-1)|\leq C e^{-\frac{1}{2}(\delta-\varepsilon)|k_i|}.
 \end{equation}
 If (\ref{Case11}) holds, by (\ref{Equ}), we also have
 \begin{equation}\label{Case13}
    |u(m_i+1)+u(m_i-1)|\leq C e^{-\frac{1}{2}(\delta-\varepsilon)|k_i|}.
 \end{equation}
 Putting (\ref{Case11})  and (\ref{Case13})  together, we get
 \begin{equation}\label{Phi1}
    ||\Phi(m_i)+\Phi_i(m_i)||\leq C e^{-\frac{1}{2}(\delta-\varepsilon)|k_i|}.
 \end{equation}

   If (\ref{Case12}) holds, we have
   \begin{equation}\label{Phi2}
    ||\Phi(m_i)-\Phi_i(m_i)||\leq C e^{-\frac{1}{2}(\delta-\varepsilon)|k_i|}.
 \end{equation}
 Thus  in case 1  there exists $\iota\in\{-1,1\}$  such that
 \begin{equation*}
    ||\Phi(m_i)+\iota\Phi_i(m_i)||\leq C e^{-\frac{1}{2}(\delta-\varepsilon)|k_i|}.
 \end{equation*}
  Let $T_i^1 $ and $T_i^2$ be the transfer matrices  with the potentials  $V$ and $V_i$ respectively, taking $\Phi(m_i),\Phi_i(m_i)$ to $\Phi(0),\Phi_i(0)$.

  By (\ref{G24}), (\ref{Eqp}),  the usual uniform upper semi-continuity and telescoping, one has
  \begin{equation*}
    ||T_i^1||, ||T_i^2||\leq C e^{ (L(E) +\varepsilon)m_i}.
  \end{equation*}
  and
  \begin{equation*}
    ||T_i^1-T_i^2||\leq C e^{ (L(E)-2\delta+\varepsilon)m_i}.
  \end{equation*}
  Then

  \begin{eqnarray}
    ||\Phi(0)+\iota\Phi_i(0)|| &=& ||T_i^1\Phi(m_i)+\iota T_i^2 \Phi_i(m_i)|| \nonumber\\
    &=& ||T_i^1\Phi(m_i)+\iota T_i^1 \Phi_i(m_i)-\iota T_i^1 \Phi_i(m_i)+\iota T_i^2 \Phi_i(m_i)|| \nonumber \\
     &\leq &  || T_i^1||||\Phi(m_i)+\iota\Phi_i(m_i)||+ ||T_i^1-T_i^2|| ||\Phi_i(m_i) ||\nonumber\\
      &\leq &  e^{ -( \delta-L(E)-\varepsilon)m_i} + e^{ (L(E)-2\delta+\varepsilon)m_i}\nonumber\\
       &\leq & e^{ -( \delta-L(E)-\varepsilon)m_i} \label{laadd}.
  \end{eqnarray}
  This implies $||\Phi(0)||-||\Phi(2m_i+1)||\to 0$.  This is impossible because  $u\in \ell^2(\Z).$

Case 2. $k_i$ is odd. Let $\tilde{m}_i=\frac{k_i-1}{2}$,
then
 \begin{equation*}
   \Phi(\tilde{m}_i+1)=\left(\begin{array}{cc}
                   u(\tilde{m}_i+1) \\ u(\tilde{m}_i)
                 \end{array}
   \right);
   \Phi_i(\tilde{m}_i+1)=\left(\begin{array}{cc}
                   u( \tilde{m}_i) \\ u(\tilde{m}_i+1)
                 \end{array}
   \right).
 \end{equation*}

 Applying  (\ref{Equ8}) with $n=\tilde{m}_i$, we have
 \begin{equation*}
    |u(\tilde{m}_i)+u(\tilde{m}_i+1)||u(\tilde{m}_i)-u(\tilde{m}_i+1)|\leq C e^{-(\delta-\varepsilon)|k_i|}.
 \end{equation*}
 This implies
 \begin{equation*}
   | u(\tilde{m}_i)+u(\tilde{m}_i+1)|\leq C e^{-\frac{1}{2}(\delta-\varepsilon)|k_i|},
 \end{equation*}
 or
 \begin{equation*}
    |u(\tilde{m}_i+1)-u(\tilde{m}_i)|\leq C e^{-\frac{1}{2}(\delta-\varepsilon)|k_i|}.
 \end{equation*}

 Thus  in case 2,  there also exists $\iota\in\{-1,1\}$  such that
 \begin{equation*}
    ||\Phi(\tilde{m}_i+1)+\iota\Phi_i(\tilde{m}_i+1)||\leq C e^{-\frac{1}{2}(\delta-\varepsilon)|k_i|}.
 \end{equation*}
Thus by the arguments of the case 1,  we can also  get a contradiction.
\end{proof}
\subsection{Lower bound on the eigenfunctions}

  Now we turn to the proof of the left  inequality in (\ref{G.Asymptotics}).
Our key argument for the lower bound is
  \begin{lemma}\label{Keylemma2}
  Suppose for some $k>0$ and $0<t<\ln\lambda$,
  \begin{equation*}
    ||2\theta+k\alpha||=e^{-t k}.
  \end{equation*}
  Then for any $\varepsilon>0$, we must have for large $k$,
  \begin{equation}\label{Equ3}
   || U(k)||\geq e^{-(\ln\lambda-t+\varepsilon)k}.
  \end{equation}

  \end{lemma}
  \begin{proof}

We  let $ \hat{\phi}(n)= \phi(k-n)$,
 $V(n)=2\lambda\cos2\pi(\theta+n\alpha)$ and $\hat{V}(n)=2\lambda\cos2\pi(\theta+(k-n)\alpha)$.
 Then by the assumption, one has for all $n\in\Z$,
 \begin{equation}\label{Eqpnew}
    |V(n)-\hat{V}(n)|\leq Ce^{-tk}.
 \end{equation}
 We also have
 \begin{equation}\label{Equnew}
 \phi(n+1)+\phi(n-1)+ V(n)u(n)=E\phi(n)
 \end{equation}
 and
  \begin{equation}\label{Equinew}
\hat{\phi}(n+1)+\hat{\phi}(n-1)+ \hat{V}(n)\tilde{\phi}(n)=E\hat{\phi}(n).
 \end{equation}
 Let
  \begin{equation*}
   \hat{U}(n)=\left(\begin{array}{cc}
                  \hat{\phi}(n) \\ \hat{\phi}(n-1)
                 \end{array}
   \right).
 \end{equation*}

Suppose for some small $\sigma>0$,
\begin{equation*}
    ||U(k)||\leq e^{-(\ln\lambda-t+\sigma)k}.
\end{equation*}
By Lemma \ref{Keylemma} and (\ref{G.new18})($k_0=k, y=n,y^\prime=2n$), we have
for any $k\leq |n|\leq Ck$,
\begin{eqnarray*}
  ||U(n)|| &\leq & e^{-|n-k|\ln\lambda} e^{\varepsilon|n|}||U(k)||+e^{-(\ln\lambda-\varepsilon)|n|} \\
  &\leq &  e^{-(\ln\lambda-\varepsilon)|n|}e^{(t-\sigma)k}.
\end{eqnarray*}
By Lemma \ref{Keylemma} again,  we have  for
$|n|\leq k$,
\begin{eqnarray*}
  ||U(n)|| &\leq & \max\{e^{-|n|\ln\lambda}, e^{-|n-k|\ln\lambda} ||U(k)|| \}e^{\varepsilon k}+e^{-(\ln\lambda-\varepsilon)|n|} \\
  &\leq & e^{-|n|\ln\lambda}e^{\varepsilon k}+ e^{-(2k-|n|)\ln\lambda }e^{(t-\sigma+\varepsilon)k}.
\end{eqnarray*}
This implies for $|n|\leq Ck$,
\begin{eqnarray*}
   |\hat{\phi}(n)|  |\phi(n)| &=&  |{\phi}(k-n)|  |\phi(n)| \\
  &\leq & e^{-(\ln\lambda-t  +\sigma-\varepsilon)k}+e^{-(\ln\lambda-\varepsilon)k}.
\end{eqnarray*}

By a standard calculation using (\ref{Eqpnew}), (\ref{Equnew}) and (\ref{Equinew}),
we have for any $|n|\leq C|k|$,
\begin{eqnarray}
   |W(\phi,\hat{\phi})(n)-W(\phi,\hat{\phi})(n-1)| &\leq & |V(n)-\hat{V}(n)||\phi(n)\hat{\phi}(n)|  \nonumber \\
  &\leq &   e^{-tk}|\phi(n)\hat{\phi}(n)|\nonumber   \\
  &\leq & e^{-(\ln\lambda+\sigma^{\prime}-\varepsilon)k} ,\nonumber
\end{eqnarray}
where $\sigma^\prime=\min\{\sigma,t\}$.
This implies for any $0<m\leq Ck$ and $|n|\leq Ck$,
\begin{eqnarray}
   |W(\phi,\hat{\phi})(n+m)-W(\phi,\hat{\phi})(n-1)| &\leq &  \sum_{j=0}^{m-1} e^{-(\ln\lambda+\sigma^\prime-\varepsilon)k}\nonumber \\
  &\leq &  e^{-(\ln\lambda+\sigma^\prime -\varepsilon )k}. \label{almostconstancy2}
\end{eqnarray}


By  (\ref{decayingupper}), for some $n_0=Ck$,
we must have
\begin{equation*}
    |\phi(n_0)|,|\phi(n_0-1)|\leq e^{-(\ln\lambda-\delta-\varepsilon)n_0}\leq e^{-(\ln\lambda+\sigma^{\prime})k}.
\end{equation*}
This implies,
\begin{equation*}
     |W(\phi,\hat{\phi})(n_0)|\leq  e^{-(\ln\lambda+\sigma^\prime)k}.
\end{equation*}
Combining with (\ref{almostconstancy2}), we must have that
\begin{equation}\label{almostconstancyn}
    |W(\phi,\hat{\phi})(n)|\leq  e^{-(\ln\lambda+\sigma^{\prime}-\varepsilon)k}
\end{equation}
holds for all $|n|\leq Ck$.

Now we split $k$ into odd or even to discuss the problem.

Case 1. $k$ is even. Let $m=\frac{k}{2}$,
then
 \begin{equation*}
  U(m)=\left(\begin{array}{cc}
                   \phi(m) \\ \phi(m-1)
                 \end{array}
   \right);
   \hat{U}(m)=\left(\begin{array}{cc}
                   {\phi}( m) \\ {\phi}(m+1)
                 \end{array}
   \right).
 \end{equation*}
 Applying  (\ref{almostconstancyn}) with $n=m-1$, we have
 \begin{equation*}
    |\phi(m)||\phi(m+1)-\phi(m-1)|\leq e^{-(\ln\lambda+\sigma^\prime-\varepsilon)k}.
 \end{equation*}
 This implies
 \begin{equation}\label{Case11new}
    |\phi(m)|\leq  e^{-\frac{1}{2}(\ln\lambda+\sigma^\prime-\varepsilon)k},
 \end{equation}
 or
 \begin{equation}\label{Case12new}
    |\phi(m+1)-\phi(m-1)|\leq  e^{-\frac{1}{2}(\ln\lambda+\sigma^\prime-\varepsilon)k}.
 \end{equation}
 If (\ref{Case11new}) holds, by (\ref{Equnew}), we also have
 \begin{equation}\label{Case13new}
    |\phi(m+1)+\phi(m-1)|\leq e^{-\frac{1}{2}(\ln\lambda+\sigma^\prime-\varepsilon)k}.
 \end{equation}
 Putting (\ref{Case11new})  and (\ref{Case13new})  together, we get
 \begin{equation}\label{Phi1new}
    ||U(m)+\hat{U}(m)||\leq e^{-\frac{1}{2}(\ln\lambda+\sigma^\prime-\varepsilon)k}.
 \end{equation}

   If (\ref{Case12new}) holds,  we have
   \begin{equation}\label{Phi2new}
    ||U(m)-\hat{U}(m)||\leq e^{-\frac{1}{2}(\ln\lambda+ \sigma^{\prime}-\varepsilon)k}.
 \end{equation}
 Thus  in case 1  there exists $\iota\in\{-1,1\}$  such that
 \begin{equation*}
    ||U(m)+\iota\hat{U}(m)||\leq e^{-\frac{1}{2}(\ln\lambda+\sigma^\prime-\varepsilon)k}.
 \end{equation*}
  Let $T $ and $\hat{T}$ be the transfer matrices associated  to potentials  $V$ and $\hat{V}$, taking $U(m),\hat{U}(m)$ to $U(0),\hat{U}(0)$ correspondingly.

  By (\ref{G24}), (\ref{Eqpnew}),
    the usual uniform upper semi-continuity and telescoping, one has
  \begin{equation*}
    ||T||, ||\hat{T}||\leq  e^{ (\ln\lambda +\varepsilon)m}.
  \end{equation*}
  and
  \begin{equation*}
    ||T-\hat{T}||\leq  e^{ (\ln\lambda-2t+\varepsilon)m}.
  \end{equation*}
  By the  right  inequality of (\ref{G.Asymptotics})($\ell=m,x_0=k$), it is easy to see that
  \begin{equation}\label{half}
   ||\hat{U}(m) ||\leq e^{-(\ln\lambda-\varepsilon)m}.
  \end{equation}

  Then, as in (\ref{laadd}), we have
  \begin{eqnarray*}
   ||U(0)+\iota \hat{U}(0)||   &\leq &  || T||||U(m)+\iota\hat{U}(m)||+ ||T-\hat{T}|| ||\hat{U}(m) ||\\
      &\leq &  e^{ (\ln\lambda+\varepsilon)m}e^{-\frac{1}{2}(\ln\lambda+\sigma^\prime-\varepsilon)k} + e^{ (\ln\lambda-2t+\varepsilon)m}e^{-m\ln\lambda }.
  \end{eqnarray*}
  This implies $||U(0)||-||{U}(2m+1)||\to 0$. This is impossible because  $\phi\in \ell^2(\Z).$

Case 2. $k$ is odd. Let $\tilde{m}=\frac{k-1}{2}$,
then
 \begin{equation*}
   U(\tilde{m}+1)=\left(\begin{array}{cc}
                   \phi(\tilde{m}+1) \\ \phi(\tilde{m})
                 \end{array}
   \right);
   \hat{U}(\tilde{m}+1)=\left(\begin{array}{cc}
                   \phi( \tilde{m}) \\ \phi(\tilde{m}+1)
                 \end{array}
   \right).
 \end{equation*}

 Combining with (\ref{almostconstancyn}), we have
 \begin{equation*}
    |\phi(\tilde{m})+\phi(\tilde{m}+1)||\phi(\tilde{m})-\phi(\tilde{m}+1)|\leq  e^{-(\ln\lambda+\sigma^\prime-\varepsilon)k}.
 \end{equation*}
 This implies
 \begin{equation*}
   | \phi(\tilde{m})+\phi(\tilde{m}+1)|\leq  e^{-\frac{1}{2}(\ln\lambda+\sigma^\prime-\varepsilon)k},
 \end{equation*}
 or
 \begin{equation*}
    |\phi(\tilde{m}+1)-\phi(\tilde{m})|\leq  e^{-\frac{1}{2}(\ln\lambda+\sigma^\prime-\varepsilon)k}.
 \end{equation*}

 Thus  in case 2,  there also exists $\iota\in\{-1,1\}$  such that
 \begin{equation*}
    ||U(\tilde{m}+1)+\iota\hat{U}(\tilde{m}+1)||\leq C e^{-\frac{1}{2}(\ln\lambda+\sigma^\prime-\varepsilon)k}.
 \end{equation*}
  As before, we also  get   a contradiction.
\end{proof}
{\bf Proof of the  left    inequality of (\ref{G.Asymptotics})}
\begin{proof}
The  left  inequality of (\ref{G.Asymptotics}) already follows  except for Subcase ii in the proof of Theorem \ref{MaintheoremAL}, by Remark \ref{Rcase1}.

Thus we only need to consider the case when $\eta\geq \gamma=\frac{\varepsilon}{C}$. Letting  $t=\eta\frac{|\ell|}{|x_0|}$ and $k=x_0$ in Lemma \ref{Keylemma2},
we obtain
\begin{equation*}
    ||U(x_0)||\geq e^{-(\ln\lambda+\varepsilon)|x_0|} e^{\eta|\ell|}.
\end{equation*}
Combining with (\ref{G.new18}), this completes the proof.
\end{proof}

\section{Universal reflective hierarchical structure}\label{refl}
We first present the local version of Theorem \ref{Maintheoremdecay}. The definition of $f(\ell)$ in Theorem \ref{Maintheoremdecay} depends on $\theta$ and $\alpha$.
Thus sometimes we will write $f_{\alpha,\theta}(\ell)$ to make clear what $\theta$ is used.

                     \begin{theorem}\label{Maintheoremdecaylocal}
                     Fix $\delta$ with $0<\delta<\ln\lambda$.
                     Suppose   $\alpha$ is  Diophantine.
                    Let $\varepsilon>0$ be small enough.
                     Then there exists $L_0=L_0(\lambda,\alpha,\delta,\hat{C})$ \footnote{We omit the dependence on $\varepsilon$ whenever  $\varepsilon$ is (implicitly) present in the
statement.} such that if for all $k$ with $  L_1\leq |k|\leq C|\ell| $,
\begin{equation}\label{2thetaalphalocal}
   ||2\theta+2s_0\alpha+k \alpha||_{\R/\Z}\geq  e^{- (\delta+\varepsilon)|k |}.
\end{equation}
and the solution of $H\phi=E\phi$  satisfies  \eqref{shn} for all $k$
with $|k-s_0|\leq C|\ell|$ and $||U(s_0)||=1$, where $C=C(\alpha,\delta,\lambda)$ is a large constant and  $L_0\leq L_1\leq \frac{|\ell|}{C}$, 
then the following statement holds:
                      \par
                      Let   $x_0$ (we can choose any one if $x_0$ is not unique) be such that
              \begin{equation*}
             |\sin\pi(2\theta+2s_0\alpha+x_0\alpha)|  = \min_{|x|\leq 2|\ell|}|\sin\pi(2\theta+2s_0\alpha+x\alpha)|.
              \end{equation*}
              Then  if $|x_0|\geq L_1$, we have
                     \begin{equation}\label{G.Asymptoticslocal}
                      f_{\alpha,\theta+s_0\alpha}(\ell)e^{-\varepsilon|\ell|} \leq ||U(\ell)||\leq f_{\alpha,\theta+s_0\alpha}(\ell)e^{\varepsilon|\ell|}.
                       \end{equation}
                       If $|x_0|\leq L_1$, we have
                     \begin{equation}\label{G.Asymptoticslocalnew}
                     e^{-\ln\lambda|\ell|}e^{-\varepsilon|\ell|} \leq ||U(\ell)||\leq e^{-\ln\lambda|\ell|} e^{\varepsilon|\ell|}.
                       \end{equation}
                     \end{theorem}
                     \begin{proof}
                     Case 1:$|x_0|\geq L_1$.
                     In sections 3 and 4, we completed the proof of
                     Theorem \ref{Maintheoremdecay}. It is immediate
                     that if we shift the operator by $s_0$ units and replace the definition of the generalized eigenfunctions
                     $\phi$ with the assumption of \eqref{shn} only on the scale $C|\ell|$, our arguments will hold for  \eqref{G.Asymptoticslocal}  directly. In order to avoid  repetition, we  omit the proof.
                     \par
                     Case 2: $|x_0|\leq L_1$.
                     \eqref{G.Asymptoticslocalnew} follows directly from  Lemma \ref{Keylemma} by
                     shifting the operator   $s_0$ units.

                     \end{proof}
                     \begin{remark}\label{re52}
                     In order to obtain \eqref{G.Asymptoticslocal}, we
                     only need  condition  \eqref{2thetaalphalocal} on
                     scale $\frac{|\ell|}{C}\leq |k|\leq C|\ell|$ and
                     condition \eqref{shn}  on scale $|k|\leq C|\ell|$. Moreover, if we assume for $|k|\leq L_1$,
                     \begin{equation}\label{2thetaalphalocalnew}
   ||2\theta+2s_0\alpha+k \alpha||_{\R/\Z}\geq  e^{- (\delta+\varepsilon)|k |},
\end{equation}
then \eqref{G.Asymptoticslocalnew} and \eqref{G.Asymptoticslocal} imply
 \begin{equation}\label{G.Asymptoticslocalnew1}
                      f_{\alpha,\theta+s_0\alpha}(\ell)e^{-\varepsilon|\ell|} \leq ||U(\ell)||\leq f_{\alpha,\theta+s_0\alpha}(\ell)e^{\varepsilon|\ell|}.
                       \end{equation}
                       in both cases.
                     \end{remark}

We will now prove Theorem \ref{Universalend}.

\begin{theorem}\label{localmaximum}
Fix $\varsigma_1>0$, $0<\delta<\ln\lambda$ and $s_0\in \Z$.
Then there exists a constant $L_0=L_0(\alpha,\lambda,\delta,\varsigma_1)$  such that the following statement holds.
Let $L_1\geq L_0$. Suppose  $K$ satisfies
 $ |K| \geq CL_1$ and
\begin{equation}\label{Kiepsilon}
   ||2\theta+2s_0\alpha+K \alpha||_{\R/\Z}\leq  e^{-\varsigma_1|K |},
\end{equation}
and  for all $k$ with $L_1\leq |k|\leq C|K|$
\begin{equation}\label{2thetaalpha}
   ||2\theta+2s_0\alpha+k \alpha||_{\R/\Z}\geq  e^{- (\delta+\varepsilon)|k |},
\end{equation}
and
$s_0$ is a $CK$-local maximum, 
where $C=C(\alpha,\lambda,\delta,\varsigma_1)$ is a large constant.
 Then  there exists a $\frac{3\varsigma_1}{4\ln\lambda}K$-local
 maximum \footnote{$3/4$ can be replaced with $1-\varepsilon$ for any $\varepsilon>0$.}$b_K$ such that
\begin{equation}\label{Kilocalold}
    |b_K-K-s_0|\leq 2L_1.
\end{equation}

\end{theorem}
\begin{proof}
By shifting the operator, we can assume $s_0=0$.
Let $\epsilon $ be such that
$$||2\theta+2s_0\alpha+K \alpha||_{\R/\Z}=  e^{-\epsilon |K|}.$$
Then $\varsigma_1\leq\epsilon\leq \delta+\varepsilon$.

  By Theorem \ref{Maintheoremdecaylocal}\footnote{$s_0$ is a local
    maximum so that $\hat{C}$ in \eqref{shn} is 1, thus the largeness
    in Theorem \ref{Maintheoremdecaylocal} does not depend on $\hat{C}.$} with
  $\ell=x_0=K$, one has
  \begin{equation}\label{Newequ1}
 e^{-(\ln\lambda-\epsilon +\varepsilon)|K |} \leq \frac{|| U(s_0+K)||}{||U(s_0)||}\leq e^{-(\ln\lambda-\epsilon -\varepsilon)|K|}.
  \end{equation}
  By Theorem \ref{Maintheoremdecaylocal} again, one has
  \begin{equation}\label{Newequ2}
   \sup_{|k|\leq \varepsilon |K |}|| U(K +k)||=\sup_{|k|\leq \frac{3\varsigma_1}{4\ln\lambda}|K|} || U(K +k)||.
  \end{equation}
  Thus there exists a $\frac{3\varsigma_1}{4\ln\lambda}|K|$-local maximum  $b_{K}$ such that
   \begin{equation}\label{Newequ3}
| b_{K}-K|\leq \varepsilon |K|.
  \end{equation}
Suppose \eqref{Kilocalold} does not hold. Then there exists
$k_0$  with $  2L_1\leq |k_0|\leq \varepsilon K $   such that
\begin{equation}\label{20171}
      || U(K +k_0)||= \sup_{|k|\leq \varepsilon |K |}|| U(K +k)||=\sup_{|k|\leq \frac{3\varsigma_1}{4\ln\lambda}|K|} || U(K +k)||.
\end{equation}
where $L_1$ is such that  (\ref{Kiepsilon}),(\ref{2thetaalpha}) hold.

Case 1. $\min_{|k|\leq 2|k_0|} ||2\theta+k\alpha||_{\R/\Z}\geq e^{-\varepsilon |k_0|}$.

Let $\frac{p_n}{q_n}$ be the continued fraction expansion of $\alpha$.
For $\gamma>0$ (we will let $\gamma=\frac{\varepsilon}{C}$),
let $n$ be the largest integer such that
\begin{equation*}
   2 q_n \leq  \gamma |k_0|,
\end{equation*}
and  let $s$ be the largest positive integer such that $2sq_n\leq
\gamma |k_0| $.

Construct intervals $I_1=[sq_n,sq_{n}-1]$ and $I_2=[K+k_0-sq_n,K+k_0+sq_{n}-1]$.

\textbf{ Claim 1:} We have
\begin{equation}\label{smallplus}
    \min_{i,i^\prime\in I_1\cup I_2}||2\theta+(i+i^\prime)\alpha||_{\R/\Z} \geq e^{-\varepsilon |k_0|}
\end{equation}
and  for any $i\neq i^\prime$, $i,i^\prime\in I_1\cup I_2$,
\begin{equation}\label{smallplus1}
    ||(i-i^\prime)\alpha||_{\R/\Z} \geq e^{-\varepsilon |k_0|}.
\end{equation}

By Theorem \ref{universalth} and the DC condition on $\alpha$, $\{\theta_i\}_{i\in I_1\cup I_2}$ is $\varepsilon$-uniform.
Combining with  Lemma  \ref{Le.Uniform}, there exists some $i_0$ with  $i_0\in I_1\cup I_2$
    such that
      $ \theta_{i _0}\notin  A_{4sq_n-1,\ln\lambda-\varepsilon }$.
      By Lemma \ref{lem}, $i_0$ cannot be in $I_1$ so must be in $I_2$.
      Set $I=[i_0-2sq_n+1,i_0+2sq_n-1]=[x_1,x_2]$.   By    (\ref{Cramer1}), (\ref{Cramer2}) and (\ref{Numerator}) again,
 one has
\begin{equation*}
|G_I(K +k_0,x_i)|\leq e^{(\ln\lambda+\varepsilon )(4sq_n-1-|K_j+k_0-x_i|)-(4sq_n-1)(\ln\lambda -\varepsilon)}\leq e^{ \varepsilon sq_n}e^{-|K+k_0-x_i|\ln \lambda }.
\end{equation*}
 Notice that $|K+k_0-x_1|,|K+k_0-x_2|\geq sq_n-1$.
 By (\ref{Block}) and (\ref{20171}),
 \begin{equation*}
    |\phi(K +k_0)|\leq e^{-(\ln\lambda-\varepsilon)sq_n } (|\phi(x_1)|+|\phi(x_0)|)\leq e^{-(\ln\lambda-\varepsilon)sq_n  } ||U(K+k_0)||.
\end{equation*}
Similarly,
\begin{equation*}
    |\phi(K +k_0-1)|\leq e^{-(\ln\lambda-\varepsilon)sq_n  } ||U(K+k_0)||.
\end{equation*}
The last two inequalities imply that
\begin{equation}\label{imp}
  ||U(K+k_0)||  \leq e^{-(\ln\lambda-\varepsilon)sq_n  } ||U(K+k_0)||.
\end{equation}
Since $2(s+1)q_n\geq \gamma |k_0|$ and  $|k_0|\geq 2L_1$,  \eqref{imp}   is impossible.  

Case 2.  $\min_{|k|\leq 2|k_0|} ||2\theta+k\alpha||_{\R/\Z}\leq e^{-\varepsilon |k_0|}$ for some $\varepsilon>0$.

In this case,   we  construct, as before,  intervals $I_1$ around $0$ and $I_2$ around $K+k_0$.

Suppose $i\in I_1$.
For $i^\prime\in I_2$, we have
\begin{eqnarray}
  \nonumber ||2\theta +(i+i^\prime)\alpha|| _{\R/\Z}&\geq & ||(i+i^\prime-K)\alpha|| _{\R/\Z}-||2\theta +K\alpha|| _{\R/\Z} \\
   &\geq &  ||(i+i^\prime-K)\alpha|| _{\R/\Z}-e^{-\epsilon |K|}\label{exchange1}
\end{eqnarray}
and
\begin{eqnarray}
 \nonumber ||(i-i^\prime)\alpha|| _{\R/\Z}&\geq & ||2\theta+(i-i^\prime+K)\alpha|| _{\R/\Z}-||-2\theta -K\alpha|| _{\R/\Z} \\
   &\geq &   ||2\theta+(i-i^\prime+K)\alpha|| _{\R/\Z}-e^{-\epsilon |K|}.\label{exchange2}
\end{eqnarray}
Suppose $i\in I_2$.

For $i^\prime\in I_1$, we have
\begin{eqnarray}
  \nonumber ||2\theta +(i+i^\prime)\alpha|| _{\R/\Z}&\geq & ||(i-K+i^\prime)\alpha|| _{\R/\Z}-||2\theta +K\alpha|| _{\R/\Z} \\
   &\geq &  ||(i-K+i^\prime)\alpha|| _{\R/\Z}-e^{-\epsilon |K|}\label{exchange5}
\end{eqnarray}
and
\begin{eqnarray}
  ||(i-i^\prime)\alpha|| _{\R/\Z} &\geq &|| 2\theta-(i-K-i^\prime)\alpha|| _{\R/\Z}-||2\theta +K\alpha|| _{\R/\Z}\\
 &\geq & || 2\theta-(i-K-i^\prime)\alpha|| _{\R/\Z}-e^{-\epsilon |K|} .\label{exchange6}
\end{eqnarray}
For $i^\prime\in I_2$, we have
\begin{eqnarray}
  \nonumber ||2\theta +(i+i^\prime)\alpha|| _{\R/\Z}&\geq & ||-2\theta+(i-K+i^\prime-K)\alpha|| _{\R/\Z}-||4\theta +2K\alpha|| _{\R/\Z} \\
   &\geq & ||2\theta-(i-K+i^\prime-K)\alpha|| _{\R/\Z}-2e^{-\epsilon |K|}\label{exchange3}
\end{eqnarray}
and
\begin{eqnarray}
  ||(i-i^\prime)\alpha|| _{\R/\Z}= || (i-K-(i^\prime-K))\alpha|| _{\R/\Z}.\label{exchange4}
\end{eqnarray}
Inequalities from (\ref{exchange1}) to (\ref{exchange4}) imply that
the small divisor
conditions on $\theta_i +\theta_{i^\prime}$ and
$\theta_i -\theta_{i^\prime}$ get swapped upon shifting of the elements in $I_2$ by $K$ units.

Let $|x_0|\leq 2|k_0|$ be such that
$||2\theta+x_0\alpha||_{\R/\Z}\leq e^{-\varepsilon |k_0|}$.

Case 2.1. $|k_0+x_0|\geq \varepsilon |k_0|$. In this case let $[x_1,x_2]=[K+k_0-\varepsilon |k_0|,K+k_0+\varepsilon |k_0|]$.
 By the small divisor conditions (\ref{exchange1}) to
 (\ref{exchange4}) and following the proof of  \eqref{first},  and \eqref{second}, we get
 \begin{eqnarray}
    ||U(K+k_0)|| &\leq& e^{-(\ln\lambda-\varepsilon) |x_1-K-k_0|}||U(x_1)||+e^{-(\ln\lambda-\varepsilon) |x_2-K-k_0|}||U(x_2)|| \\
    &\leq&  e^{-(\ln\lambda-\varepsilon) \varepsilon |k_0|}||U(x_1)||+e^{-(\ln\lambda-\varepsilon) \varepsilon|k_0|}||U(x_2)||\\
    &\leq& e^{-(\ln\lambda-\varepsilon) \varepsilon |k_0|}||U(K+k_0)||,\label{imp2}
 \end{eqnarray}
 where the third inequality holds because $K+k_0$ is the local maximum. \eqref{imp2} is also impossible for $|k_0|\geq 2L_1$.


Case 2.2. $|k_0+x_0|\leq \varepsilon |k_0|$.
In this case, $|x_0|\geq  \frac{1}{2} |k_0|\geq   L_1$ so that the condition \eqref{2thetaalpha} holds for all $|k|\geq |x_0| $.
  By the  small divisor conditions (\ref{exchange1}) to (\ref{exchange4}) again, and following the proof of (\ref{Equ9}), we get (using \eqref{20171})
  \begin{eqnarray*}
   ||U(K+k_0)|| 
      &\leq& ||U(K+k_0)||e^{-(\ln\lambda -\delta-\varepsilon) |k_0|}.
  \end{eqnarray*}
 This is also impossible. 

\end{proof}
\textbf{Proof of Claim 1.}
\begin{proof}
Without loss of generality assume $i\in I_1$.
For $i^\prime\in I_2$, by the DC condition on $\alpha$ we have
\begin{eqnarray*}
  ||2\theta +(i+i^\prime)\alpha|| _{\R/\Z}&\geq & ||(i+i^\prime-K)\alpha|| _{\R/\Z}-||2\theta +K\alpha|| _{\R/\Z} \\
   &\geq & e^{-\varepsilon |k_0|}
\end{eqnarray*}
and
\begin{eqnarray*}
  ||(i-i^\prime)\alpha|| _{\R/\Z}&\geq & ||2\theta+(i-i^\prime+K)\alpha|| _{\R/\Z}-||-2\theta -K\alpha|| _{\R/\Z} \\
   &\geq & e^{-\varepsilon |k_0|}.
\end{eqnarray*}
For $i^\prime\in I_1$, the proof is trivial.
\end{proof}
{\bf Proof of Theorem  \ref{Universalend}}
\begin{proof}
Without loss of generality, assume $k_0=0$.
Let $\hat{K}=L_0(\alpha,\lambda,\delta,
\varsigma)$
in Theorem \ref{localmaximum}.

 By Theorem \ref{localmaximum} with $s_0=0$,    $K=K_{j_0}$, $\varsigma_1= \varsigma $ and $L_1=\hat{K}$,
there exists a local $\frac{3\varsigma}{4\ln\lambda}K_{j_0}$ maximum $b_{{j_0}}$ such that
  $| b_{{j_0}}-K_{j_0}|\leq 2\hat{K}$.
  Let $ b_{{j_0}}-K_{j_0}=b_{{j_0}}^\prime$ with $|b_{{j_0}}^\prime|\leq 2\hat{K}$.

  Shifting the operator $H_{\lambda,\alpha,\theta}$  by $b_{{j_0}}$ units, we get the operator $H_{\lambda,\alpha,\theta+b_{{j_0}}\alpha}$.
  By the conditions of Theorem \ref{Universalend},
  $\zeta<\delta+\varepsilon<\ln \lambda,$  we have
  \begin{eqnarray}
  \nonumber  ||2(\theta+b_{{j_0}}\alpha)+k\alpha||_{\R/\Z} &\geq& ||2\theta-(2b_{{j_0}}^\prime\alpha+k\alpha)||_{\R/\Z} -||4\theta+2K_{j_0}\alpha||_{\R/\Z} \\
   \nonumber  &\geq&  ||2\theta-(2b_{{j_0}}^\prime\alpha+k\alpha)||_{\R/\Z} - 2e^{-(\varsigma +\varepsilon)|K_{j_0}|}\\
      \nonumber &\geq& e^{- (\delta+\varepsilon)(|k | +4\hat{K})}\\
                 &\geq& e^{- (\delta+\varepsilon) |k |},\label{smallshift}
  \end{eqnarray}
  for all $\frac{1}{2}\hat{K}^2\leq |k|\leq \frac{ \varsigma}{\ln\lambda} |K_{j_0}|$.
  Similarly, we have
  \begin{eqnarray}
  \nonumber  ||2(\theta+b_{{j_0}}\alpha) +(-K_{j_1}-2b_{{j_0}}^\prime)\alpha||_{\R/\Z} &\leq& ||2\theta+K_{j_1}\alpha||_{\R/\Z} +||4\theta+2K_{j_0}\alpha||_{\R/\Z} \\
   \nonumber  &\leq&  ||2\theta+K_{j_1}\alpha||_{\R/\Z} - 2e^{-(\varsigma +\varepsilon)|K_{j_0}|}\\
                 &\leq& e^{- \frac{3}{4} \varsigma |-K_{j_1}-2b_{{j_0}}^\prime |},\label{smallshift1}
  \end{eqnarray}

  By Theorem \ref{localmaximum} with $s_0=b_{{j_0}}$,    $K=-K_{j_1}-2b_{{j_0}}^\prime$, $\varsigma_1=\frac{3}{4} \varsigma $  and $L_1=\frac{1}{2}\hat{K}^2$, we get
  there exists a local $\frac{9\varsigma}{16\ln\lambda}K_{j-1}$ maximum $b_{K_{j_0},K_{j_1}}$ such that
  $| b_{{j_0},{j_1}}-b_{{j_0}}-(-K_{j_1}-2b_{{j_0}}^\prime)|\leq \hat{K}^2$.
  This implies $ b_{{j_0},{j_1}}=K_{j_0}-b_{{j_0}}^\prime-K_{j_1}+b_{{j_1}}^\prime$ with $|b_{{j_1}}^\prime|\leq \hat{K}^2$.

 Shifting the operator $H_{\lambda,\alpha,\theta}$  by $b_{{j_0},{j_1}}$ units, we get the operator $H_{\lambda,\alpha,\theta+b_{{j_0},{j_1}}\alpha}$.
  Thus we have
  \begin{eqnarray}
  \nonumber  ||2(\theta+b_{{j_0},{j_1}}\alpha)+k\alpha||_{\R/\Z} &\geq& ||2\theta- 2b_{{j_0}}^\prime\alpha+2b_{{j_1}}^\prime \alpha+k\alpha)||_{\R/\Z} -2||2\theta+K_{j_0}\alpha||_{\R/\Z} -2||2\theta+K_{j_1}\alpha||_{\R/\Z} \\
   \nonumber  &\geq&  ||2\theta+(- 2b_{K_{j_0}}^\prime+2b_{{j_1}}^\prime +k)\alpha)||_{\R/\Z} - 4e^{-(\varsigma +\varepsilon)|K_{j_1}|}\\
      \nonumber &\geq& e^{- (\delta+\varepsilon)(|k | +2\hat{K}+2\hat{K}^2)}\\
                 &\geq& e^{- (\delta+\varepsilon) |k |},\label{smallshift2}
  \end{eqnarray}
  for all $\frac{1}{2}(\hat{K}+\hat{K}^2)\hat{K}\leq |k|\leq \frac{ \varsigma}{\ln\lambda} |K_{j-1}|$.
  Similarly, we have
  \begin{eqnarray}
  \nonumber  ||2(\theta+b_{{j_0},{j_1}}\alpha) +(K_{j-2}+2b_{K_j}^\prime-2b_{{j_1}}^\prime)\alpha||_{\R/\Z} &\leq& ||2\theta+K_{j_2}\alpha||_{\R/\Z} + 2e^{-(\varsigma +\varepsilon)|K_{j_1}|} \\
                 &\leq& e^{- \frac{3}{4} \varsigma |K_{j_2}+2b_{{j_0}}^\prime-2b_{{j_1}}^\prime |}.
  \end{eqnarray}

  By Theorem \ref{localmaximum} with $s_0=b_{{j_0},{j_1}}$,
  $K=K_{j_2}+2b_{{j_0}}^\prime-2b_{{j_1}}^\prime$,
  $\varsigma_1=\frac{3}{4} \varsigma $ and
  $L_1=\frac{1}{2}(\hat{K}^2+\hat{K}^3)$, we get that
  there exists a local $\frac{9\varsigma}{16\ln\lambda}K_{j_2}$ maximum $b_{{j_0},{j_1},{j_2}}$ such that
  $ b_{{j_0},{j_1},{j_2}}=K_{j_0}+b_{{j_0}}^\prime-K_{j_1}-b_{{j_1}}^\prime+K_{j_2}+b_{{j_2}}^\prime$ with $|b_{{j_2}}^\prime|\leq \hat{K}^2+\hat{K}^3$.

Define $a_{n}=\hat{K}^2(\hat{K}+1)^{n-2}$  for $n\geq2$  and $a_1=\hat{K}$.
Then $a_n=\hat{K}\sum_{i=1}^{n-1} a_i$.
 Notice that by (\ref{KiKi-1})
  \begin{equation}\label{induction1}
    \sum_{i=0}^s ||2\theta+K_{j_i}\alpha||_{\R/\Z}\leq  \sum_{i=0}^s e^{-(\varsigma +\varepsilon)|K_{j_i}|}\leq 2e^{-(\varsigma +\varepsilon)|K_{j_s}|}.
  \end{equation}
 We will  prove that for any $1\leq s\leq k$
  there exists a local $\frac{9\varsigma}{16\ln\lambda}K_{j_s}$ maximum $b_{{j_0},{j_1},\cdots,{j_s}}$ such that
  \begin{equation}\label{induction2}
   b_{{j_0},{j_1},\cdots,{j_s}}=\sum_{i=0}^s (-1)^iK_{j_i}+(-1)^{i-s}b_{{j_i}}^\prime
  \end{equation}
  with $|b_{{j_i}}^\prime|\leq a_{i+1}$ by induction in $s.$

Assume that \eqref{induction2} holds for $s$. We will prove that \eqref{induction2} holds for $s+1$.

 Shifting the operator $H_{\lambda,\alpha,\theta}$  by $b_{{j_0},{j_1},\cdots,{j_s}}$ units, we get the operator $H_{\lambda,\alpha,\theta+b_{{j_0},{j_1},\cdots,{j_s}}\alpha}$.
  Arguing as in  (\ref{smallshift2}) we have
  \begin{equation*}
    ||2(\theta+b_{{j_0},{j_1},\cdots,{j_s}}\alpha)+k\alpha||_{\R/\Z}\;\;\;\;\;\;\;\;\;\;\;\;\;\;\;\;\;\;\;\;\;\;\;\;\;\;\;\;\;\;\;\;\;\;\;\;\;\;\;\;\;\;
  \end{equation*}
\begin{equation*}
    \;\;\;\;\;\;\;\;\;\;\;\;\;\;\;\;\;\;\;\;\;\;\;\;\;\;\;\;\geq ||2\theta+(2\sum_{i=0}^{s}(-1)^{i+1} b_{{j_i}}^\prime)\alpha +(-1)^{s+1}k\alpha||_{\R/\Z} - 2\sum_{i=0}^{s}||2\theta+K_{j_i}\alpha||_{\R/\Z}
  \end{equation*}
  \begin{equation}\label{shift1}
    \;\;\;\;\;\;\;\;\;\;\;\;\;\;\;\;\;\;\;\;\;\;\;\;\;\;\;\;\geq ||2\theta+(2\sum_{i=0}^{s}(-1)^{i+1} b_{{j_i}}^\prime)\alpha +(-1)^{s+1}k\alpha||_{\R/\Z} - 2\sum_{i=0}^{s}e^{-(\varsigma +\varepsilon)|K_{j_i}|}
  \end{equation}
  \begin{equation*}
   \geq  e^{- (\delta+\varepsilon)(|k | +2\sum_{i=1}^{s+1}a_i)} \;\;\; \;\;\;\;\;\;\;\;\;\;\;\;\;\;\;\;\;\;\;\;\;\;\;\;\;\;\;\;\;\;\;\;\;\;\;\;\;\;\;\;\;\;
  \end{equation*}
  \begin{equation}\label{smallshift3}
    \geq e^{- (\delta+\varepsilon) |k |}\;\;\;\;\;\;\;\;\;\;\;\;\;\;\;\;\;\;\;\;\;\;\;\;\;\;\;\;\;\;\;\;\;\;\;\;\;\;\;\;\;\;\;\;\;\;\;\;\;\;\;\;\;\;\;\;\;\;\;\;\;\;
  \end{equation}
  for all $\frac{1}{2}a_{s+2}\leq |k|\leq  \frac{ \varsigma}{\ln\lambda} |K_{j_s}|$, since $\sum_{i=1}^{s+1}a_i=\frac{1}{\hat{K}}a_{s+2}$.
  Similarly to (\ref{smallshift1}), we have
   $$||2(\theta+b_{{j_0},{j_1},\cdots,{j_s}}\alpha) +((-1)^{s+1}K_{j_{s+1}}+2\sum_{i=0}^{s}(-1)^{s+i+1} b_{{j_i}}^\prime)\alpha||_{\R/\Z} $$
   \begin{equation*}
   \leq ||2\theta+K_{j_{s+1}}\alpha||_{\R/\Z} +4e^{-(\varsigma +\varepsilon)|K_{j_1}|}
 \end{equation*}
 \begin{equation}\label{smallshift44}
           \;\;\;\;\;\;  \leq    e^{- \frac{3}{4} \varsigma |(-1)^{s+1} K_{j_{s+1}}+2\sum_{i=0}^{s}(-1)^{s+i+1} b_{{j_i}}^\prime |}.
\end{equation}

  By Theorem \ref{localmaximum} with $s_0=b_{{j_0},{j_1},\cdots,{j_s}}$,    $K=(-1)^{s+1}K_{j_{s+1}}+2\sum_{i=0}^{s}(-1)^{s+i+1} b_{{j_i}}^\prime$, $\varsigma_1=\frac{3}{4} \varsigma $ and $L_1=\frac{1}{2}a_{s+2}$, we get
  that there exists a local $\frac{9\varsigma}{16\ln\lambda}K_{j_{s+1}}$ maximum $b_{{j_0},{j_1},\cdots,{j_{s+1}}}$ such that
  \begin{eqnarray*}
    b_{{j_0},{j_1},\cdots,{j_{s+1}}} &=& b_{{j_0},{j_1},\cdots,{j_s}}+(-1)^{s+1}K_{j_{s+1}}+2\sum_{i=0}^{s}(-1)^{s+i+1} b_{{j_i}}^\prime+b_{{j_{s+1}}}^\prime \\
      &=& \sum_{i=0}^{s-1} (-1)^iK_{j_i}+(-1)^{i-s-1}b_{{j_i}}^\prime
  \end{eqnarray*}
  with $|b_{{j_{s+1}}}^\prime|\leq a_{s+2}$.

 By the fact
  \begin{eqnarray*}
     |b_{{ j_0},{ {j_1}},...,{ {j_s}}}-
  \sum_{i=0}^s (-1)^{i} K_{j_i}| &\leq& \sum_{i=0}^s   |b_{{j_i}}^\prime |\\
    &\leq&  \sum_{i=1}^{s+1} a_i\\
     &\leq& ( \hat{K}+1)^{s+1} ,
  \end{eqnarray*}
  we complete the proof of  I of Theorem  \ref{Universalend}.

Now we start to prove  II of Theorem  \ref{Universalend}.
Fix some $0\leq s\leq k$.
Let us consider a local $\frac{9\varsigma}{16\ln\lambda}K_{j_s}$
maximum $b_{{j_0},{j_1},\cdots,{j_s}}$ and shift the operator by $b_{{j_0},{j_1},\cdots,{j_s}}$
units. We get the operator $H_{\lambda,\alpha,\theta+b_{{j_0},{j_1},\cdots,{j_s}}\alpha}$.
As in \eqref{shift1}, we also have
\begin{equation*}
    ||2(\theta+b_{{j_0},{j_1},\cdots,{j_s}}\alpha)+k\alpha||_{\R/\Z}\;\;\;\;\;\;\;\;\;\;\;\;\;\;\;\;\;\;\;\;\;\;\;\;\;\;\;\;\;\;\;\;\;\;\;\;\;\;\;\;\;\;
  \end{equation*}
\begin{equation}\label{notchange}
    \;\;\;\;\;\;\;\;\;\;\;\;\;\;\;\;\;\;\;\;\;\;\;\;\;\;\;\;\leq ||2\theta+(2\sum_{i=0}^{s}(-1)^{i+1} b_{{j_i}}^\prime)\alpha +(-1)^{s+1}k\alpha||_{\R/\Z} + 2\sum_{i=0}^{s}e^{-(\varsigma +\varepsilon)|K_{j_i}|}
  \end{equation}
for all $a_{s+2}\leq |k|\leq  \frac{ \varsigma}{\ln\lambda} |K_{j_s}|$.

Actually, the definition of $f(\ell)$ in Theorems \ref{Maintheoremdecay} and \ref{Maintheoremdecaylocal}, depends on $\theta$ and $\alpha$.
Thus we will use $f_{\alpha,\theta}(\ell)$ with $|\ell|\geq Ca_{s+2}$.
 Let   $\ell_0$  be such that
              \begin{equation*}
             |\sin\pi(2\theta+2b_{{j_0},{j_1},\cdots,{j_s}}\alpha+\ell_0\alpha)|  = \min_{|x|\leq 2|\ell|}|\sin\pi(2\theta +2b_{{j_0},{j_1},\cdots,{j_s}}\alpha+x\alpha)|.
              \end{equation*}
By \eqref{shift1} and \eqref{notchange}, we have for $|\ell_0|\geq a_{s+2}$,
\begin{equation}\label{notchange1}
  e^{-\varepsilon |\ell|} f_{\alpha,\theta}((-1)^{s+1}\ell) \leq f_{\alpha,\theta+b_{{j_0},{j_1},\cdots,{j_s}}\alpha}(\ell)\leq f_{\alpha,\theta}((-1)^{s+1}\ell) e^{\varepsilon |\ell|}.
\end{equation}
If $|\ell_0|\leq a_{s+2}$, we have
\begin{equation}\label{ellcase}
   e^{-\varepsilon |\ell|} e^{-\ln\lambda |\ell|} \leq f_{\alpha,\theta}(\ell)\leq  e^{-\ln\lambda |\ell|}e^{\varepsilon |\ell|},
\end{equation}
since $|\ell|\geq C a_{s+2}$.

Let $x_s=x-b_{{j_0},{j_1},\cdots,{j_s}}$.
If $|x_s| \in[Ca_{s+2},\frac{1}{C}\frac{ \varsigma}{\ln\lambda} |K_{j_s}|]$, II of Theorem
\ref{Universalend} follows from
  Theorem \ref{Maintheoremdecaylocal} and \eqref{notchange1}  and \eqref{ellcase}.

  If $|x_s| \in[\frac{1}{C}\frac{ \varsigma}{\ln\lambda} |K_{j_s}|,\frac{ \varsigma}{4\ln\lambda} |K_{j_s}|]$, II of Theorem
\ref{Universalend} follows from
  Lemma \ref{Keylemma} and the fact that  $b_{{j_0},{j_1},\cdots,{j_s}}$ is  a local $\frac{9\varsigma}{16\ln\lambda}K_{j_s}$ maximum. Notice that in this case
\begin{equation*}
   e^{-\varepsilon |x_s|} e^{-\ln\lambda |x_s|} \leq f_{\alpha,\theta}((-1)^{s+1}x_s)\leq  e^{-\ln\lambda |x_s|}e^{\varepsilon |x_s|}.
\end{equation*}

\end{proof}

\section{Asymptotics of the transfer matrices}\label{transfer}

{\bf    Proof of Theorem \ref{Thtransferasy}  }

Proof.
Without loss of generality, we consider $\ell>0$. First assume $x_0<0$ or $\eta\leq \gamma=\frac{\varepsilon}{C}.$
By Theorem \ref{Maintheoremdecay}, in those cases,   
 one has
\begin{equation*}
    ||U(\ell)||\leq e^{-(\ln\lambda-\varepsilon)\ell}.
\end{equation*}
By (\ref{G.new17}),
we have
\begin{equation*}
    ||A_{\ell}||\geq ||U(\ell)||^{-1}\geq e^{(\ln\lambda-\varepsilon)\ell}.
\end{equation*}
Combining with (\ref{G24}), the proof follows.

Now we turn to the proof of  the case when $x_0>0$ and $\eta>\gamma$.  We will assume $\ell>0$ is large enough.  By  (\ref{largex0}), one has  $x_0>0$ is large enough.
Thus below we always assume $x_0$ is large.
\begin{theorem}\label{Th.new2transfer}
Assume $j x_0\leq k<(j+1)x_0$ with $k\geq \frac{x_0}{8}$, where $j=0,1$.  Then
we have
\begin{equation}\label{G.newnewnew1transfernew}
  ||A_k||\leq \max\{e^{-|k-jx_0|\ln\lambda}||A_{jx_0}||,e^{-|k-(j+1)x_0|\ln\lambda}||A_{(j+1)x_0}||\}e^{\varepsilon k},
\end{equation}
\begin{equation}\label{G.newnewnew2transfernew}
  ||A_k||\geq \max\{e^{-|k-jx_0|\ln\lambda}||A_{jx_0}||,e^{-|k-(j+1)x_0|\ln\lambda}||A_{(j+1)x_0}||\}e^{-\varepsilon k}.
\end{equation}
\end{theorem}
\begin{proof}

  Apply (\ref{second}) with $k_0=x_0, y=k$  $y^\prime=2x_0$ and $\varphi=\psi$.
  We have for $j x_0\leq k<(j+1)x_0$ with $k\geq \frac{x_0}{8}$,
  \begin{equation}\label{Equ1}
    ||\tilde{U}(k)||\leq \max\{e^{-|k-jx_0|\ln\lambda}||\tilde{U}{(jx_0)}||,e^{-|k-(j+1)x_0|\ln\lambda}||\tilde{U}{((j+1)x_0})||\}e^{\varepsilon k}
  \end{equation}

By Last-Simon's arguments ((8.6) in \cite{last1999eigenfunctions}), one has
\begin{equation}\label{G.last}
||A_k||\geq  ||A_k\tilde{{U}} (0)||\geq c ||A_k|| .
\end{equation}
Then (\ref{G.newnewnew1transfernew}) holds by (\ref{G.last}) and (\ref{Equ1}).

(\ref{G.newnewnew2transfernew}) holds directly by (\ref{G24}).
\end{proof}
\begin{lemma}\label{Keylemma3}
For any $2x_0\leq k\leq Cx_0$,

\begin{equation*}
e^{-\varepsilon x_0} ||A_{2x_0}||e^{\ln\lambda |k-2x_0|}\leq||A_{k}||\leq e^{\varepsilon x_0} ||A_{2x_0}||e^{\ln\lambda |k-2x_0|}
\end{equation*}
\end{lemma}
\begin{proof}
The right inequality holds directly. It suffices to show the left inequality.

By (\ref{Equ9}) and  noting $t\leq \delta+\varepsilon$,
we have
\begin{equation*}
    ||\tilde{U} (x_0)||\leq \max\{e^{-(\ln \lambda-\delta-\varepsilon)x_0 } ||\tilde{U} (0)||,  e^{-(\ln \lambda-\delta-\varepsilon)x_0 } ||\tilde{U} (2x_0)|| \}
\end{equation*}
Clearly,
 $||\tilde{U} (x_0)|| \leq e^{-(\ln \lambda-\delta-\varepsilon)x_0 } ||\tilde{U} (0)||$ can not happen. Otherwise, by the fact
 $ || {U} (x_0)||\leq e^{-(\ln \lambda-\delta-\varepsilon)x_0 } ||{U} (0)||$, we must have
 \begin{equation*}
    |\phi(x_0)\psi(x_0-1)-\phi(x_0-1)\psi(x_0)|\leq e^{-(\ln \lambda-\delta-\varepsilon)x_0}.
 \end{equation*}
 This   contradicts (\ref{W}).

Thus we must have
\begin{equation}\label{Equ2}
    ||\tilde{U} (x_0)||\leq   e^{-(\ln \lambda-\delta-\varepsilon)x_0 } ||\tilde{U} (2x_0)||.
\end{equation}
Lemma holds directly if $k\leq 2x_0+\frac{\varepsilon}{C}x_0$.
If  $ k-2x_0 \geq \frac{\varepsilon}{C}x_0$,
 by (\ref{second}) again ($k_0=x_0, y=2x_0, y^{\prime}=k,\gamma=\frac{\varepsilon}{C}$), one has
\begin{equation*}
    ||\tilde{U} (2x_0)||\leq \max\{e^{-(\ln \lambda-\varepsilon)x_0 } ||\tilde{U} (x_0)||, e^{-(\ln \lambda-\varepsilon)|k-2x_0| } ||\tilde{U} (k)||\}.
\end{equation*}
Combining with (\ref{Equ2}), we must have
\begin{equation*}
   ||\tilde{U} (k)||\geq e^{(\ln \lambda-\varepsilon)|k-2x_0| }  ||\tilde{U} (2x_0)||.
\end{equation*}
Combining with (\ref{G.last}),
 we get the left inequality.

\end{proof}
\begin{lemma}\label{Keylemma4}
The following holds
\begin{equation}\label{Gx0}
   e^{(\ln\lambda-\varepsilon)x_0}\leq || A_{x_0}|| \leq e^{(\ln\lambda+\varepsilon)x_0},
\end{equation}
\begin{equation}\label{G2x0}
   e^{(\ln\lambda-\varepsilon)2x_0}e^{-\eta\ell}\leq || A_{2x_0}|| \leq e^{(\ln\lambda+\varepsilon)2x_0}e^{-\eta\ell}.
\end{equation}
\end{lemma}
\begin{proof}
We first prove (\ref{Gx0}).  The right inequality holds  by (\ref{G24}) directly. Thus it suffices to show the left one.
By (\ref{second}), for any $\frac{x_0}{8}\leq k< x_0$, one has
\begin{equation*}
    ||U(k)||\leq \max\{e^{-k\ln\lambda},e^{-|k-x_0|\ln\lambda} ||U(x_0)||\} e^{\varepsilon k}.
\end{equation*}

 Clearly
\begin{equation}\label{G.transfer2}
||A_k||\geq ||U(k)||^{-1},
\end{equation}
thus  by  (\ref{G.newnewnew1transfernew}),
we must have for any $\frac{x_0}{8}\leq k< x_0$,
\begin{equation}\label{G.transfer3}
  \max\{e^{-k\ln\lambda} ,e^{-|k-x_0|\ln\lambda}||A_{x_0}||\}e^{\varepsilon k}\geq(\max\{e^{-k\ln\lambda},e^{-|k-x_0|\ln\lambda} ||U(x_0)||\})^{-1}e^{-\varepsilon k}.
\end{equation}
Recall that by (\ref{x0}) and (\ref{Equ3}),
\begin{equation}\label{Equ10}
  e^{-(\ln\lambda-\eta^\prime +\varepsilon)x_0} \leq  ||U(x_0)||\leq e^{-(\ln\lambda-\eta^\prime -\varepsilon)x_0},
\end{equation}
where $\eta^{\prime}=\frac{\ell}{x_0}\eta$.
Let
\begin{equation*}
  k_0=x_0- \frac{\eta^\prime }{2\ln\lambda} x_0.
\end{equation*}
One has $k_0\geq \frac{x_0}{2}$,
thus by (\ref{Equ10})
\begin{equation*}
\max\{e^{-k_0\ln\lambda},e^{-|k_0-x_0|\ln\lambda}     || U({x_0})||\}\leq  e^{-(\ln\lambda-\frac{\eta^\prime}{2})x_0} e^{\varepsilon k_0}.
\end{equation*}
Combining with (\ref{G.transfer3}), one has
\begin{equation*}
 \max\{e^{-k_0\ln\lambda} ,e^{-|k_0-x_0|\ln\lambda}||A_{x_0}||\}\geq  e^{(\ln\lambda-\frac{\eta^\prime}{2})x_0} e^{-\varepsilon k_0}.
\end{equation*}
This  implies
\begin{equation*}
 ||A_{x_0}||\geq e^{(\ln \lambda-\varepsilon)x_0}.
\end{equation*}
Now we start to prove (\ref{G2x0}).
By (\ref{G.Asymptotics})($\ell=2x_0$), one has
\begin{equation}\label{Equ5}
  e^{-(\ln\lambda+\varepsilon)2x_0}e^{\eta ^{\prime}x_0}\leq  || U(2x_0)||\leq e^{-(\ln\lambda-\varepsilon)2x_0}e^{\eta ^{\prime}x_0},
\end{equation}
Combining with (\ref{G.transfer2}),
one has
\begin{equation*}
   || A_{2x_0}||\geq e^{(\ln\lambda-\varepsilon)2x_0}e^{-\eta ^{\prime}x_0}.
\end{equation*}
Thus it remains  to prove the right inequality of (\ref{G2x0}).
By (8.5) and (8.7) in \cite{last1999eigenfunctions}
we have
\begin{equation}\label{G.transfer8}
   ||A_kU(0)||^2\leq ||A_k||^2m(k)^2+||A_k||^{-2},
\end{equation}
where
\begin{equation}\label{G.transfer9}
  m(k)\leq C\sum_{p=k}^{\infty}\frac{1}{||A_p||^2}.
\end{equation}
If  $ k\geq Cx_0 $ ($C$ may depend on $\ln\lambda,\delta$),
 by Theorem \ref{Maintheoremdecay} we have
\begin{eqnarray}
  ||A_k|| &\geq & ||U(k)|| ^{-1}\nonumber\\
   &\geq & e^{(\ln\lambda-\delta-\varepsilon)k}\label{laadd1}
\end{eqnarray}
and by   (\ref{G24}) we have
\begin{equation*}
  ||A_{2x_0}||\leq e^{(\ln\lambda+\varepsilon)2x_0}.
\end{equation*}
Combining with (\ref{laadd1}), we have
\begin{equation}\label{G.lastsimon2}
  ||A_k||\geq ||A_{2x_0}||e^{\frac{\ln\lambda-\delta}{2}k}.
\end{equation}
for $k\geq C x_0$, where  $C$ is large enough.

If $2x_0\leq k\leq Cx_0$, by Lemma \ref{Keylemma3},  we have
\begin{equation}\label{Equ4}
   || A_k||\geq ||A_{2x_0}|| e^{(\ln\lambda-\varepsilon)|k-2x_0|} e^{-\varepsilon x_0}.
\end{equation}

Thus by (\ref{G.transfer9}), (\ref{G.lastsimon2}) and (\ref{Equ4}), we have
\begin{equation}\label{G.transfer10}
  m(2x_0)\leq ||A_{2x_0}||^{-2}e^{\varepsilon x_0}.
\end{equation}
Let $k=2x_0$ in (\ref{G.transfer8}). Then
\begin{equation*}
 || U(2x_0)||\leq \frac{e^{\varepsilon x_0}}{||A_{2x_0}||}.
\end{equation*}
Thus by (\ref{Equ5}), we obtain
\begin{equation}\label{G.transfer11}
 ||A_{2x_0}||\leq e^{(2\ln\lambda-\eta^\prime-\varepsilon )x_0}.
\end{equation}
\end{proof}

  Theorem \ref{Thtransferasy} for  the remaining case ($\eta\geq \gamma=\frac{\varepsilon}{C}$ and $x_0>0$) now follows
  directly from  Theorem \ref{Th.new2transfer} and Lemma \ref{Keylemma3},\ref{Keylemma4}.

{\bf Proof of Corollary \ref{C.anysolution}}
\begin{proof}
The Corollary follows from Theorem \ref{Thtransferasy} and (\ref{G.last}).
\end{proof}
{\bf Proof of Corollary \ref{densityco}}
\begin{proof}
i) and ii) of
 Corollary \ref{densityco} follow from  Theorem \ref{Thtransferasy}  and Corollary \ref{C.anysolution} directly.

Fix some small  $\varepsilon_1,\varepsilon_2>0$.
By the definition of $\delta$, there exists a sequence $n_j$ (assume $n_j>0$ for simplicity)
such that
\begin{equation*}
 e^{- (\delta+\varepsilon_1)n_j}   \leq ||2\theta+n_j\alpha||_{\R/\Z}\leq e^{-\frac{\delta}{2}n_j}.
\end{equation*}
 By the Diophantine condition on $\alpha$, we have
 \begin{equation*}
    n_{j+1}\geq e^{\frac{n_j}{C}}.
 \end{equation*}

We prove (\ref{densityE}) first.
By  Theorem \ref{Maintheoremdecay},
one has for any $ |k|\in[\varepsilon_2n_{j+1}, \frac{n_{j+1}}{2}]$,
\begin{equation*}
    ||U(k)||\leq e^{-(\ln\lambda -\varepsilon_1)|k|}.
\end{equation*}
This implies (\ref{densityE}) by the arbitrariness  of $\varepsilon_1,\varepsilon_2$.

Now we turn to the proof of (\ref{densityT}).
By  Theorem \ref{Thtransferasy},
one has for any $ |k|\in[\varepsilon_2n_{j+1}, n_{j+1}]$,
\begin{equation*}
    ||A_k||\geq  e^{(\ln\lambda -\varepsilon_1)|k|}.
\end{equation*}
This implies (\ref{densityT}).
\end{proof}
\appendix
\section{Uniformity}
The following lemma is critical when we prove Theorem \ref{universalth}.
\begin{lemma} (\text{Lemma } 9.7, \cite{MR2521117})\label{a1}
Let $\alpha\in \mathbb{R}\backslash \mathbb{Q}$, $x\in\mathbb{R}$ and $0\leq k_0 \leq q_n-1$ be such that
$ |\sin\pi(x+k_0\alpha)|=\inf_{0\leq k\leq q_n-1}    |\sin\pi(x+k \alpha)|$, then for some absolute constant $C > 0$,
\begin{equation}\label{G927}
    -C\ln q_n\leq \sum _{k=0,k\neq k_0}^{q_n-1} \ln|\sin\pi(x+k\alpha )|+(q_n-1)\ln2\leq  C\ln q_n.
\end{equation}
  \end{lemma}
{\bf Proof of Theorem \ref{universalth}}
\begin{proof}
Let $i_0,j_0\in I_1\cup I_2$  be such that $|\sin\pi (2\theta+(i_0+j_0)\alpha)|=\min_{i,j\in I_1\cup I_2}|\sin\pi (2\theta+(i+j)\alpha)|$.
By the Diophantine condition on $\alpha$, there exist
$\tau^\prime,\kappa^\prime>0$ such that for any $i+j\neq  i_0+j_0$ and $i, j\in I_1\cup I_2$,
\begin{equation}\label{appsmall1}
   |\sin\pi(2\theta+(i+j)\alpha)|\geq \frac{\tau^\prime}{(sq_n)^{\kappa^\prime}}.
 \end{equation}
 Also for all
 $i,j\in I_1\cup I_2,i\neq j$, we have
 \begin{equation}\label{appsmall2}
   |\sin\pi( j-i)\alpha|\geq \frac{\tau^\prime}{(sq_n)^{\kappa^\prime}}.
 \end{equation}
 In (\ref{Def.Uniform}), let $x=\cos2\pi a$, $k=sq_{n}-1$  and
take    the logarithm. Then
  $$ \ln \prod_{ j\in I_1\cup  I_2  , j\neq i } \frac{|\cos2\pi a-\cos2\pi\theta_j|}
        {|\cos2\pi\theta_i-\cos2\pi\theta_j|}\;\;\;\;\;\;\;\;\;\;\;\;\;\;\;\;\;\;\;\;\;\;\;\;\;\;\;\;\;\;\;\;\;\;\;\;\;\;\;\;\;\;\;\;\;\;\;\;\;\;\;\;\;
        \;\;\;\;\;\;\;\;\;\;\;\;\;\;\;$$
          $$=   \sum _{ j\in I_1\cup  I_2  , j\neq i }\ln|\cos2\pi a-\cos2\pi \theta_j|- \sum _{ j\in I_1\cup  I_2  , j\neq i }\ln|\cos2\pi\theta_i  -\cos2\pi \theta_j| .$$

First,  we  estimate $ \sum _{ j\in I_1\cup  I_2  , j\neq i }\ln|\cos2\pi a-\cos2\pi \theta_j| $.
Obviously,
   $$ \sum _{ j\in I_1\cup  I_2  , j\neq i }\ln|\cos2\pi a-\cos2\pi \theta_j| \;\;\;\;\;\;\;\;\;\;\;\;\;\;\;\;\;\;\;\;\;\;\;\;\;\;\;\;\;\;\;\;\;\;\;\;\;\;\;\;\;\;\;\;\;\;\;\;\;\;\;\;\;\;\;\;$$
$$\;\;\;\;\;\;\;\;\;\;\;\;\;\;\;\;\;\;\;\;\;=\sum_{ j\in I_1\cup  I_2  , j\neq i }\ln|\sin\pi(a+\theta_j)|+\sum_{ j\in I_1\cup  I_2  , j\neq i }\ln |\sin\pi(a-\theta_j)|
+(sq_{n}-1)\ln2  $$
\begin{equation*}
    =\Sigma_{+}+\Sigma_-+(sq_{n}-1)\ln2.  \;\;\;\;\;\;\;\;\;\;\;\;\;\;\;\;\;\;\;\;\; \;\;\;\;\;\;\;\;\;\;\;\;\;\;\;\;\;\;\;\;\;\;\;\;\;\;\;\;\;\;\;\;
\end{equation*}
Both $\Sigma_+$ and $\Sigma_-$ consist of $s$ terms of the form of  (\ref{G927}), plus s terms of the form
\begin{equation*}
    \ln\min_{j=0,1,\cdots,q_{n}}|\sin\pi(x+j\alpha)|,
\end{equation*}
minus $\ln|\sin\pi(a\pm\theta_i)|$.
       Thus,  using   (\ref{G927})  s times for $\Sigma_{+}$ and $\Sigma_{-}$ respectively,  one has
 \begin{equation}\label{G.appnumerator}
   \sum_{j \in I_1  \cup I_2,j\neq i}\ln|\cos2\pi a-\cos2\pi \theta_{j}|\leq-sq_{n}\ln2+Cs\ln q_{n}.
\end{equation}
If $a=\theta_i$,
we obtain
   $$ \sum _{j \in I_1  \cup I_2,j\neq i}\ln|\cos2\pi \theta_i-\cos2\pi \theta_j| \;\;\;\;\;\;\;\;\;\;\;\;\;\;\;\;\;\;\;\;\;\;\;\;\;\;\;\;\;\;\;\;\;\;\;\;\;\;\;\;\;\;\;\;\;\;\;\;\;\;\;\;\;\;\;\;$$
$$\;\;\;\;\;\;\;\;\;\;\;\;\;\;\;\;\;\;\;\;\;=\sum_{j \in I_1  \cup I_2,j\neq i}\ln|\sin\pi(\theta_i+\theta_j)|+\sum_{j \in I_1  \cup I_2,j\neq i}\ln |\sin\pi(\theta_i-\theta_j)|
+(sq_{n}-1)\ln2  $$
\begin{equation}\label{G.appsumdenumerate}
    =\Sigma_{+}+\Sigma_-+(sq_{n}-1)\ln2,  \;\;\;\;\;\;\;\;\;\;\;\;\;\;\;\;\;\;\;\;\;\;\;\; \;\;\;\;\;\;\;\;\;\;\;\;\;\;\;\;\;\;\;\;\;\;\;\;\;\;\;\;\;\;\;
\end{equation}
 where
  \begin{equation*}
   \Sigma_{+}=\sum_{j \in I_1  \cup I_2,j\neq i}\ln |\sin\pi(2\theta+ (i+j) \alpha)|,
 \end{equation*}
 and
\begin{equation*}
     \Sigma_-=\sum_{j \in I_1  \cup I_2,j\neq i}\ln |\sin\pi( i-j)\alpha|.
 \end{equation*}
  We will estimate $\Sigma_+ $.
 Set
$J_1=[1,s_1]$ and
$J_2=[s_1+1,s]$, which are two adjacent disjoint intervals of length
$s_1,s_2$ respectively.
 Then $I_1\cup
I_2$ can be represented as a disjoint union of segments $B_j,\;j\in
J_1\cup J_2,$ each of length $q_{n}$.
Applying (\ref{G927}) to  each  $B_j$, we
obtain
\begin{equation}\label{G312}
\Sigma_+ \geq -sq_{n}\ln 2+
\sum_{j\in J_1\cup J_2 }\ln
|\sin  \pi\hat \theta_j|-Cs\ln q_{n}-\ln|\sin2\pi (\theta+i\alpha)|,
\end{equation}
where
\begin{equation}\label{G313}
|\sin  \pi\hat \theta_j|=\min_{\ell \in B_j}|\sin  \pi
(2\theta +(\ell +  i)\alpha )|.
\end{equation}
By (\ref{appG1}) and (\ref{appsmall1}), we have
\begin{equation}\label{appG2}
   \sum_{j\in J_1\cup J_2 }\ln
|\sin  \pi\hat \theta_j|\geq -\gamma s q_n -Cs\ln sq_n
\end{equation}
Putting (\ref{appG2}) in (\ref{G312}), we get
\begin{equation}\label{appG3}
    \Sigma_+ \geq -sq_{n}\ln 2 -\gamma s q_n -Cs\ln sq_n.
\end{equation}
Similarly, replacing (\ref{appG1}), (\ref{appsmall1}) with (\ref{appsmall2}),
 and  arguing as in  the proof of (\ref{appG3}), we obtain,
\begin{equation}\label{newG323}
\Sigma_- > -sq_{n}\ln 2  -Cs\ln sq_n.
\end{equation}

From  (\ref{G.appsumdenumerate}), (\ref{appG3}) and (\ref{newG323}), one has
$$ \sum _{j \in I_1  \cup I_2,j\neq i}\ln|\cos2\pi \theta_i-\cos2\pi \theta_j| \;\;\;\;\;\;\;\;\;\;\;\;\;\;\;\;\;\;\;\;\;\;\;\;\;\;\;\;\;\;\;\;\;\;\;\;\;\;\;\;\;\;\;\;\;\;\;\;\;\;\;\;\;\;\;\;$$
\begin{equation}\label{G.app325}
  \geq -sq_{n}\ln 2 -\gamma s q_n -Cs\ln sq_n.
\end{equation}
By (\ref{G.appnumerator})  and (\ref{G.app325}),
we have
\begin{equation*}
        \max_{ i\in I_1\cup I_2} \prod_{j \in I_1\cup I_2,j \neq i } \frac{|x-\cos2\pi\theta_{j }|}
        {|\cos2\pi\theta_i-\cos2\pi\theta_{j }|}<e^{ sq_{n}(\gamma+C\frac{\ln sq_n}{q_n}  )  }.
      \end{equation*}
      By the assumption $s\leq q_n^C$, we get for any $\varepsilon>0$ and large $n$,
      \begin{equation*}
        \max_{ i\in I_1\cup I_2} \prod_{j \in I_1\cup I_2,j \neq i } \frac{|x-\cos2\pi\theta_{j }|}
        {|\cos2\pi\theta_i-\cos2\pi\theta_{j }|}<e^{ sq_{n}(\gamma+  \varepsilon  )  }.
      \end{equation*}
      This completes the proof.
\end{proof}
\section{Block Expansion Theorem}
{\bf Proof of Theorem \ref{blockth}}
\begin{proof}

\par
   For any $\hat{y}  \in [y_1+\gamma k,y_2-
\gamma k]$, by the assumption we have
 there exists an interval $ I(\hat{y})=[x_1,x_2]\subset
[ y_1,y_2]$
such that $\hat{y}\in I(\hat{y})$ with $\frac{\gamma}{20}k\leq|I(\hat{y})|\leq \frac{1}{2}\text{dist }(y,\{y_1,y_2\})$, and
\begin{equation}\label{G329}
    \text{dist}(\hat{y},\partial I(\hat{y}))\geq  \frac{1}{40} |I(\hat{y})| \geq\frac{\gamma}{800}k
\end{equation}
and
\begin{equation}\label{G330}
  |G_{I(\hat{y})}(\hat{y},x_i)| \leq e^{-\tau|\hat{y}-x_i|},\;i=1,2,
\end{equation}
where
 $ \{x_1,x_2\}=\partial I(\hat{y})$ is the boundary of the interval $I(\hat{y})$.
   For $z  \in  \partial I(\hat{y})$,  let
  $z' $ be the neighbor of $z$, (i.e., $|z-z'|=1$) not belonging to $I(\hat{y})$.
\par
If $x_2+1\leq y_2- \gamma k$ or  $x_1-1\geq  y_1+\gamma k$,
we can expand $\varphi(x_2+1)$ or $\varphi(x_1-1)$ using (\ref{Block}). We can continue this process until we arrive to $z$
such that $z+1>y_2- \gamma k$ or  $z-1<  y_1+\gamma k$, or the  numbers of iterations  reach
$\lfloor\frac{160 0}{ \gamma}\rfloor$. Then, by (\ref{Block})
\begin{equation}\label{G331}
   \varphi(y)=\displaystyle\sum_{s ; z_{i+1}\in\partial I(z_i^\prime)}
G_{I(y)}(k,z_1) G_{I(z_1^\prime)}
(z_1^\prime,z_2)\cdots G_{I(z_s^\prime)}
(z_s^\prime,z_{s+1})\varphi(z_{s+1}^\prime),
\end{equation}
where in each term of the summation one has
$y_1+\gamma k+1\leq z_i\leq  y_2-\gamma k-1$, $i=1,\cdots,s,$ and
  either $z_{s+1} \notin [y_1+\gamma k+1,y_2+\gamma k-1]$, $s+1 < \lfloor\frac{1600}{ \gamma}\rfloor$; or
$s+1= \lfloor \frac{1600}{ \gamma}\rfloor$.
We should mention that $z_{s+1}\in[y_1,y_2] $.
\par
 If $z_{s+1} \in [y_1,y_1+ \gamma k]$, $s+1 <  \lfloor\frac{1600}{ \gamma}\rfloor$,
this implies
\begin{equation*}
    |\varphi(z_{s+1}^\prime)|\leq r_{y_1}^{\varphi}.
\end{equation*}

By  (\ref{G330}), we have for such terms
\begin{equation*}
    \nonumber
   | G_{I( {y})}( {y},z_1) G_{I(z_1^\prime)}
(z_1^\prime,z_2)\cdots G_{I(z_s^\prime)}
(z_s^\prime,z_{s+1})\varphi(z_{s+1}^\prime)|
\end{equation*}
\begin{eqnarray}
\nonumber
&\leq & r_{y_1} ^{\varphi}e^{-\tau(|y-z_1|+\sum_{i=1}^{s}|z_i^\prime-z_{i+1}|)}
 \\
\nonumber
&\leq & r_{y_1} ^{\varphi}e^{-\tau(|y-z_{s+1}|-(s+1))}  \\
&\leq & r_{y_1} ^{\varphi} e^{-\tau(|y-y_1|- \gamma  k -\frac{ 1600}{ \gamma})}.
\label{G332}
\end{eqnarray}
 If $z_{s+1} \in [y_2-\gamma k,y_2 ]$, $s+1 < \lfloor\frac{1600}{ \gamma}\rfloor$,
 by the same arguments,
  we have
  \begin{equation}\label{G.addGreen}
    | G_{I(y)}(y,z_1) G_{I(z_1^\prime)}
(z_1^\prime,z_2)\cdots G_{I(z_s^\prime)}
(z_s^\prime,z_{s+1})\varphi(z_{s+1}^\prime)|\leq  r_{y_2} ^{\varphi}e^{-\tau(|y-y_2|- \gamma k -\frac{ 1600}{ \gamma})}.
  \end{equation}
  If $s+1= \lfloor\frac{1600}{ \gamma}\rfloor,$
using   (\ref{G329}) and (\ref{G330}), we obtain
\begin{equation}\label{G333}
     | G_{I(y)}(y,z_1) G_{I(z_1^\prime)}
(z_1^\prime,z_2)\cdots G_{I(z_s^\prime)}
(z_s^\prime,z_{s+1})\varphi(z_{s+1}^\prime)|\leq e^{-\tau \frac{\gamma}{800} k \lfloor\frac{1600}{ \gamma}\rfloor}|\varphi(z_{s+1}^\prime)| .
\end{equation}
Notice that the total number of terms in (\ref{G331})
is  at most  $2^{\lfloor\frac{1600}{ \gamma}\rfloor}$ and $ |y-y_1|,|y-y_2|\geq 10 \gamma k$.  By (\ref{G332}), (\ref{G.addGreen}) and (\ref{G333}),  we have
\begin{equation}\label{G.add1}
|\varphi(y)|\leq  \max\{r_{y_1} ^{\varphi}e^{-\tau (|y-y_1|-3\gamma k) }, r_{y_2}^{\varphi} e^{-\tau (|y-y_2|-3\gamma k) }, \max_{p\in[y_1,y_2]}\{e^{-\tau k }|\varphi(p)|\}\}.
\end{equation}
Now we will show that for any $p\in[y_1,y_2]$, one has
$ |\varphi(p)|\leq  \max\{ r_{y_1}^{\varphi},r_{y_2}^{\varphi}\}$. Then (\ref{G.add1}) implies Theorem \ref{blockth}.
Otherwise, by the definition of $r_{y_1}^{\varphi}$ and $r_{y_2}^{\varphi}$, if   $|\varphi(p^\prime)|$ is the largest one of $|\varphi(z)|,z\in [ y_1+10\gamma k +1,y_2-10\gamma k-1]$,
then $|\varphi(p^\prime)|>\max\{ r_{y_1}^{\varphi},r_{y_2}^{\varphi}\}$. Applying (\ref{G.add1}) to $\varphi(p^\prime) $ and noticing  that  $ |p^\prime-y_1|,|p^\prime-y_2|\geq 10 \gamma k$,
we get
\begin{equation*}
|\varphi(p^\prime)|\leq  \max\{  e^{-7\tau \gamma k } r_{y_1}^{\varphi}, e^{-7\tau \gamma k } r_{y_2}^{\varphi},e^{-\tau k}|\varphi(p^\prime)|\}.
\end{equation*}
This is impossible  because $|\varphi(p^\prime)|>\max\{ r_{y_1}^{\varphi},r_{y_2}^{\varphi}\}$.

\end{proof}
 \section*{Acknowledgments}
The
work of S.J. was supported by  the Simons Foundation.  W.L. was supported by the AMS-Simons Travel Grant 2016-2018. This research was
partially
 supported by NSF DMS-1401204 and NSF DMS-1700314.
  We are grateful to the Isaac
    Newton Institute for Mathematical Sciences, Cambridge, for its
    hospitality, supported by EPSRC Grant Number EP/K032208/1, during the programme Periodic and Ergodic Spectral
    Problems where this work was started.


\footnotesize
 \bibliographystyle{abbrv} 

\end{document}